\documentclass{article}
\usepackage{fullpage}
\usepackage{amsmath}
\usepackage{amsthm}
\usepackage{amssymb}
\usepackage{url}
\usepackage{graphicx}
\usepackage{verbatim}
\usepackage{listings}
\usepackage{url}
\usepackage{graphicx}
\usepackage{tikz}
\usetikzlibrary{calc}
\usetikzlibrary{patterns}
\tikzstyle{mybox} = [draw=black, fill=white,  thick,
    rectangle, inner sep=10pt, inner ysep=20pt]
\tikzstyle{mybox} = [draw=black, fill=white,  thick,
    rectangle, inner sep=2pt, inner ysep=2pt]

\newtheorem{thm}{Theorem}
\newtheorem{lemma}{Lemma}
\newtheorem{cor}{Corollary}
\newtheorem{prop}{Proposition}
\theoremstyle{definition}
\newtheorem{definition}{Definition}
\newtheorem{remark}{Remark}
\newtheorem{example}{Example}

\begin{document}


\title{A Characterization Theorem and An Algorithm for A Convex Hull Problem}
\author{Bahman Kalantari \\
Department of Computer Science, Rutgers University, NJ\\
kalantari@cs.rutgers.edu
}
\date{}
\maketitle



\begin{abstract}
Given $S= \{v_1, \dots, v_n\} \subset \mathbb{R} ^m$ and $p \in \mathbb{R} ^m$, testing if $p \in conv(S)$, the convex hull of $S$, is a fundamental problem in computational geometry and linear programming. First, we prove a Euclidean {\it distance duality}, distinct from  classical separation theorems such as Farkas Lemma: $p$ lies in $conv(S)$ if and only if for each $p' \in conv(S)$ there exists a {\it pivot}, $v_j \in S$ satisfying $d(p',v_j) \geq d(p,v_j)$.  Equivalently, $p \not \in conv(S)$ if and only if there exists a {\it witness}, $p' \in conv(S)$ whose Voronoi cell relative to $p$ contains $S$. A witness separates $p$ from $conv(S)$ and approximate $d(p, conv(S))$ to within a factor of two. Next, we describe the {\it Triangle Algorithm}: given $\epsilon \in (0,1)$, an {\it iterate}, $p' \in conv(S)$, and $v \in S$, if $d(p, p') < \epsilon d(p,v)$, it stops. Otherwise, if there exists a pivot $v_j$, it replace $v$ with $v_j$ and $p'$ with the projection of $p$ onto the line $p'v_j$.  Repeating this process, the algorithm terminates in $O(mn  \min \{ \epsilon^{-2}, c^{-1}\ln \epsilon^{-1} \})$ arithmetic operations,  where $c$ is the {\it visibility factor}, a constant satisfying $c \geq \epsilon^2$ and $\sin (\angle pp'v_j) \leq 1/\sqrt{1+c}$, over all iterates $p'$. In particular, the geometry of the input data may result in efficient complexities such as $O(mn \sqrt[t]{\epsilon^{-2}}
 \ln \epsilon^{-1})$, $t$ a natural number, or even $O(mn \ln \epsilon^{-1})$. Additionally, (i) we prove a {\it strict distance duality} and a related minimax theorem,
 resulting in more effective pivots;  (ii) describe $O(mn \ln \epsilon^{-1})$-time algorithms that may compute a witness or a good approximate solution; (iii) prove  {\it generalized distance duality} and describe a corresponding generalized Triangle Algorithm; (iv) prove a {\it sensitivity theorem} to analyze the complexity of solving LP feasibility via the Triangle Algorithm. The Triangle Algorithm is practical and competitive  with the simplex method, sparse greedy approximation and first-order methods. Finally, we discuss future work on applications and generalizations.
\end{abstract}


\section{Introduction} \label{sec1}

Given a set $S= \{v_1, \dots, v_n\} \subset \mathbb{R} ^m$ and a distinguished point $p \in \mathbb{R} ^m$, we consider the problem of testing if $p$ lies in  $conv(S)$, the convex hull of $S$.  Throughout the article we shall refer to this problem as the  {\it convex hull decision problem}, or simply as problem (P). The convex hull decision problem is  a basic problem in computational geometry and a very special case of the {\it convex hull problem}, a problem that according to  Goodman and O'Rourke \cite{Goodman}, is a ``catch-all phrase for computing various descriptions of a polytope that is either specified as the convex hull of a finite point set or the intersection of a finite number of halfspaces.'' The descriptions include those of  vertices, facets, and  adjacencies.

Problem (P) is not only a fundamental problem in computational geometry, but in linear programming (LP). This can be argued on different grounds. On the one hand problem (P) is a very special case of LP.  On the other hand, it is well known that by the LP duality theory, the general LP problem may be cast as a single LP feasibility problem, see e.g. Chv\'atal \cite {chvatal2}.  The LP feasibility problem can then be converted into problem (P) via several different approaches. To argue the significance of (P) in linear programming, it can be justified that the two most famous polynomial-time LP algorithms, the ellipsoid algorithm of Khachiyan  \cite{kha79} and the projective algorithm of Karmarkar \cite{kar84}, are in fact explicitly or implicitly designed to  solve a case of problem (P) where $p=0$, see \cite{jinkal}. Furthermore, using an approach suggested by  Chv\'atal,  in \cite{jinkal} it is  shown that there is a direct connection between a general LP feasibility and this homogeneous case of problem (P), over integer or real input data. For integer inputs all known polynomial-time LP algorithms - when applied to  solve problem (P) - would exhibit a theoretical complexity that is polynomial in $m$, $n$, and the size of encoding of the input data, often denoted by $L$, see e.g. \cite{kha79}. The number $L$ can generally be taken to be dependent on the logarithm of $mn$ and of the logarithm of the absolute value of the largest entry in the input data. These are sometimes known as {\it weakly polynomial-time} algorithms. No {\it strongly polynomial-time} algorithm is known for LP, i.e. an algorithm that would in particular solve problem (P) in time complexity polynomial in $m$ and $n$.

Problem (P) finds applications in computational geometry itself, in particular among the variations of the convex hull problem. One such a variation is the {\it irredundancy problem}, the problem of computing all the vertices of  $conv(S)$, see \cite{Goodman}.  Clearly, any algorithm for LP can be used to solve the irredundancy problem by solving a sequence of $O(n)$ convex hull decision problems.  Clarkson \cite{clark94}, has given a more efficient algorithm than the straightforward approach of solving $n$ linear programs. The complexity of his algorithm depends on the number of vertices of the convex hull.  Furthermore, by using an approach originally due to  Matou{\v s}ek \cite{mat93} for solving closely related linear programming problems,  Chan \cite{chan96} has given a faster algorithm for the irredundancy problem.  What is generally considered to be the convex hull problem is a problem more complicated than (P) and the irredundancy problem, requiring the description of  $conv(S)$ in terms of its vertices, facets and adjacencies,  see  Chazelle \cite{chaz93} who offers an optimal algorithm.

Problem (P) can also be formulated as the minimization of a convex quadratic function over a simplex. This particular convex program  has found  applications in statistics, approximation theory, and machine learning,  see e.g Clarkson \cite{clark2008} and Zhang \cite{zhang} who consider the analysis of a greedy algorithm for the more general problem of minimizing certain convex functions over a simplex (equivalently, maximizing  concave functions over a simplex). The oldest version of such greedy algorithms is Frank-Wolfe algorithm, \cite{Frank}.   Special cases of the problem include  support vector
machines (SVM), approximating functions as convex combinations
of other functions, see e.g. Clarkson \cite{clark2008}. The problem of computing the closest point of the convex hull of a set of points to the origin, known as {\it polytope distance} is related to the case of problem (P) where $p$ is the origin. In some applications the  polytope distance refers to the distance  between two convex hulls. Gilbert's algorithm \cite{Gilbert} for the polytope distance problem is one of the earliest known algorithms for the problem. G{\"a}rtner and Jaggi \cite{Gartner} show Gilbert's algorithm coincides with Frank-Wolfe algorithm when applied to the minimization of a convex quadratic function over a simplex. In this case the algorithm is known as {\it sparse greedy approximation}.  For many other results regarding the applications of the minimization of a quadratic function over a simplex, see the bibliographies in \cite{zhang}, \cite{clark2008} and \cite{Gartner}. Clarkson \cite{clark2008} gives a through analysis of Frank-Wolfe and its variations. Another class of algorithms that are applicable to the convex hull problem are the so-called {\it first-order} methods, see
Nesterov \cite{Nesterov2} and Lan et al \cite {Lan}.

Despite all previous approaches and achievements in solving problem (P), various formulations and algorithms, investigation of this fundamental problem will most likely continue in the future. The present article focuses on solving problem (P). It reveals new and interesting geometric properties of this fundamental convex hull problem, including characterization theorems, leading to new separating hyperplane theorems, called {\it distance duality} theorems.  It then utilizes the distance dualities to describe the  {\it Triangle Algorithm}.  We analyze  theoretical complexity of the Triangle Algorithm showing that it is very attractive for solving problem (P). Additionally, utilizing the Triangle Algorithm, together by proving a sensitivity theorem, we describe a new algorithm for the LP feasibility problem and analyze its complexity.

The article is a substantially extended version of \cite{kal12}. It is organized as follows. In Section \ref{sec2} we give a detailed overview of the Triangle Algorithm, its properties and applications.  Specifically, in \ref{subsec2.1} we describe the geometry of the Triangle Algorithm and introduce two geometric problems that can be viewed as problems that are dual to (P).  In \ref{subsec2.2} we describe three complexity bounds for the Triangle Algorithm in terms of problem parameters. In  \ref{subsec2.3} we describe the complexity of the Triangle Algorithm for solving LP feasibility and optimization problems.
In \ref{subsec2.4}-\ref{subsec2.6} we give a comparison of the complexity bound of the Triangle Algorithm to that of several alternate algorithms that can solve the convex hull decision problem, namely, the simplex method, sparse greedy approximation, and first-order methods.
In \ref{subsec2.7} we describe connections between problem (P) and the general LP feasibility and optimization problems, as well as citation of some of the vast and relevant literature on LP. In \ref{subsec2.8}, we give the outline of results in the remaining sections.

\section{Overview of Triangle Algorithm, Properties and Applications} \label{sec2}

Given $u, v \in \mathbb{R} ^m$, the Euclidean distance is
$d(u,v)= \sqrt{\sum_{i=1}^m (u_i-v_i)^2} = \Vert u-v \Vert$. The Triangle Algorithm takes as input a subset $S=\{v_1, \dots, v_n\} \subset \mathbb{R} ^m$ and a point $p \in \mathbb{R} ^m$, as well as a tolerance $\epsilon \in (0,1)$.

\begin{center}
\begin{tikzpicture}
\node [mybox] (box){%
    \begin{minipage}{0.9\textwidth}
{\bf  Triangle Algorithm ($S=\{v_1, \dots, v_n\}$, $p$, $\epsilon \in (0,1))$}\

\begin{itemize}
\item
{\bf Step 0.} ({\bf Initialization}) Let $p'=v={\rm argmin}\{ d(p,v_i): v_i \in S\}$.

\item
{\bf Step 1.} If $d(p,p') < \epsilon d(p,v)$, output $p'$ as $\epsilon$-{\it approximate solution}, stop. Otherwise, if there exists $v_j \in S$ where $d(p', v_j) \geq  d(p, v_j)$, call $v_j$ a ${\it pivot}$ (or  {\it $p$-pivot}); replace $v$ with $v_j$. If no pivot exists, then output $p'$ as a {\it witness} (or {\it $p$-witness}), stop.

\item
{\bf Step 2.}  Given $p'=\sum_{i=1}^n \alpha_i v_i$, $\sum_{i=1}^n \alpha_i=1, \alpha_i \geq 0$, compute the new iterate $p''=\sum_{i=1}^n \alpha'_i v_i$, $\sum_{i=1}^n \alpha'_i=1, \alpha'_i \geq 0$ as the nearest point to $p$ on the line segment $p'v_j$.  Replace $p'$ with $p''$, $\alpha_i$ with $\alpha'_i$, for all $i$, go to Step 1.
\end{itemize}

    \end{minipage}};
\end{tikzpicture}
\end{center}

The justification in the name of the algorithm lies in the fact that in each iteration the algorithm searches for a triangle $\triangle pp'v_j$ where $v_j \in S$, $p' \in conv(S)$, such that $d(p',v_j) \geq d(p,v_j)$. Given that such triangle exists, it uses $v_j$ as a pivot to ``pull'' the current iterate $p'$ closer to $p$ to get a new iterate $p'' \in conv(S)$. The Triangle Algorithm either computes an $\epsilon$-approximate solution  $p' \in conv(S)$ satisfying
\begin{equation} \label{c1}
d(p', p) < \epsilon R, \quad R = \max \{d(v_j,p): v_j \in S\},
\end{equation}
or a witness $p' \in conv(S)$ satisfying
\begin{equation} \label{c2}
d(p',v_i) < d(p, v_i), \quad \forall i=1, \dots,n.
\end{equation}

\begin{remark} The Triangle Algorithm seeks to find an approximate solution $p' \in conv(S)$ whose relative error, $d(p',p)/R$, is within prescribed error $\epsilon$, see
(\ref{c1}), as opposed to computing a point within a prescribed absolute error.  For this reason we make three important points.

(i)  Stating the quality of approximation in terms of the relative error (\ref{c1})  is a more meaningful measure of approximation than approximation to within prescribed absolute error. Consider the case where the diameter of $conv(S)$ is very small, say less than $\epsilon$ and $p \not \in conv(S)$, but $d(p,S) < \epsilon$.

(ii) Clearly approximate solutions with a prescribed absolute error can also be computed, simply by replacing $\epsilon$ with $\epsilon/R$. This only affects the complexity by a constant factor.

(iii) All existing algorithms for the convex hull decision problem that claim to produce an approximate solution with absolute error $\epsilon$  would necessarily have a complexity bound in terms of some parameters dependent on the data. In particular, this is the case if we solve the convex hull problem decision problem via polynomial-time LP algorithms, or approximation schemes such as sparse greedy approximation, or first-order methods (see (\ref{subsec2.4}-\ref{subsec2.6})).  Analogous to polynomial-time LP algorithms, for rational inputs $S$ and $p$, when $\epsilon$ is sufficiently small the existence of an $\epsilon$-approximate solution implies $p$ lies in $conv(S)$.
\end{remark}

We refer to a point $p'$ satisfying (\ref{c2}) as $p$-{\it witness} or simply {\it witness} because this condition holds if and only if $p \not \in conv(S)$. In this case we prove the Voronoi cell of $p'$ with respect to the two point set $\{p,p'\}$ contains $conv(S)$ (see Figure \ref{Fig1}).  Equivalently,  the hyperplane that orthogonally bisects the line segment $pp'$ separates $p$ from $conv(S)$.
The set $W_p$ of all such witnesses is the intersection of $conv(S)$ and the open balls, $B_i= \{x \in \mathbb{R} ^m: \quad  d(x,v_i) < r_i \}$, $i=1, \dots, n$. $W_p$ is a convex subset of $conv(S)$ (see Figure \ref{Fig3}). It has the same dimension as $conv(S)$.

The correctness of the Triangle Algorithm lies in a new duality (theorem of the alternative) we call {\it distance duality}:

\begin{center}
\begin{tikzpicture}
\node [mybox] (box){%
    \begin{minipage}{.95\textwidth}
{\bf Distance Duality}\

Precisely one of the two conditions is satisfied:\

(i): For each $p' \in conv(S)$, there exists $v_j \in S $ such that $d(p', v_j) \geq  d(p, v_j)$;

(ii): There exists $p' \in conv(S)$ such that
$d(p',v_i) < d(p, v_i)$, for all $i=1, \dots,n$.

    \end{minipage}};
\end{tikzpicture}
\end{center}

The first condition is valid if and only if $p \in conv(S)$, and the second condition if and only if $p \not \in conv(S)$.  From the description of the Triangle Algorithm we see that given a point $p' \in conv (S)$ that is not a witness, having $d(p,p')$ as the current {\it gap}, the Triangle Algorithm moves to a new point $p'' \in  conv (S)$ where the new gap $d(p,p'')$ is reduced.

\subsection{The Geometry of the Triangle Algorithm} \label{subsec2.1}
To arrive at the distance duality mentioned above, we first prove a characterization theorem that leads to this theorem of the alternative for problem (P). We remark here that the distance duality theorem is distinct from the classical Farkas lemma, or Gordan theorem.  To arrive at this theorem we first prove:

$p \in conv(S)$ if and only if given any point $p' \in conv(S) \setminus \{p\}$, there exists $v_j$ such that $d(p',v_j) >  d(p,v_j)$.

Next, we show the strict inequality can be replaced  with $d(p',v_j) \geq d(p,v_j)$.  Thus the contrapositive theorem is:

$p \not \in conv(S)$ if and only if there exists $p' \in conv(S)$ such that $d(p',v_i) < d(p,v_i)$, for all $i=1, \dots, n$.

These two results together imply the distance duality.  A corollary of our characterization theorem reveals a geometric property of a set of balls. To describe this property, denote an open ball by $B$, its closure $\overline B$, and its boundary by $\partial B$, i.e.
\begin{equation}
B= \{x \in \mathbb{R} ^m:  \quad d(x,v) < r \}, \quad \overline B= \{x \in \mathbb{R} ^m:  \quad d(x,v) \leq r \}, \quad \partial B= \{x \in \mathbb{R} ^m:  \quad d(x,v) = r \}.\end{equation}

Consider a set of open balls $B_i= \{x \in \mathbb{R} ^m: \quad  d(x,v_i) < r_i \}$, $i=1, \dots, n$,  and let $S=\{v_1, \dots, v_n\}$.

\begin{center}
\begin{tikzpicture}
\node [mybox] (box){%
    \begin{minipage}{0.8\textwidth}
{\bf  Intersecting Balls Property:}\

If $p \in \cap_{i=1}^n \partial B_i$, then
\begin{equation}
p \in conv(S) \iff (\cap_{i=1}^n B_i) \cap conv(S)= \emptyset \iff (\cap_{i=1}^n \overline B_i) \cap conv(S)= \{p\}.\end{equation}
    \end{minipage}};
\end{tikzpicture}
\end{center}

In words, suppose a set of open balls have a common boundary point $p$. Then $p$ lies in the convex hull of their centers, if and only if the intersection of the open balls is empty, if and only if $p$ is the only point in the intersection of the closure of the balls.  A depiction of this property for a triangle is given in Figure \ref{Fig2}. This property suggests we can define a geometric {\it dual} for problem (P):

\begin{center}
\begin{tikzpicture}
\node [mybox] (box){%
    \begin{minipage}{0.8\textwidth}
{\bf  Problem (Q) (Intersecting Balls Problem):}\

Suppose there exists $p \in \cap_{i=1}^n \partial B_i$. Determine if $(\cap_{i=1}^n B_i) \cap conv(S)$ is nonempty.

    \end{minipage}};
\end{tikzpicture}
\end{center}

In fact the intersecting balls problem can be stated in more generality:

\begin{center}
\begin{tikzpicture}
\node [mybox] (box){%
    \begin{minipage}{0.8\textwidth}
{\bf  Problem (Q$'$) (General Intersecting Balls Problem):}\

Suppose there exists $p \in \cap_{i=1}^n \partial B_i$. Determine if $(\cap_{i=1}^n B_i)$ is nonempty.

    \end{minipage}};
\end{tikzpicture}
\end{center}

The Triangle Algorithm results in a fully polynomial-time approximation scheme for solving problems (Q) and (Q$'$) with the same time complexity as that of solving  problem (P).

When a point $p$ lies in $conv(S)  \cap (\cap_{i=1}^n \partial B_i)$, the union of the balls, $\cup_{i=1}^nB_i$ is referred as the {\it forbidden zone} of the convex hull of the centers, see  \cite{DKK1}  or \cite{DKK2} for the definition and some of its properties. The notion of forbidden zone of a convex set is  significant and intrinsic in the characterization of the so-called {\it mollified zone diagrams},  a variation of  {\it zone diagram} of a finite set of points in the Euclidean plane, see \cite{DKK2}. The notion of zone diagram, introduced by Asano et al \cite{asano07}, is itself a very rich and interesting variation of the classical Voronoi diagram, see e.g. \cite{arun91}, \cite{prep}.  Forbidden zones help give a characterization of mollified zone diagrams, in particular a zone diagram, \cite{DKK2}.  For  some geometric properties of forbidden zones of polygons and polytopes, see \cite{isvd12}.

\subsection{Complexity Bounds for The Triangle Algorithm} \label{subsec2.2}

The  Triangle Algorithm is geometric in nature, simple in description, and very easy to implement. Its complexity analysis also uses geometric ideas.  We derive three different complexity bounds on the number of arithmetic operations of the Triangle Algorithm.  The  first bound on the number of arithmetic operations is
\begin{equation} \label{firstbound11}
 48 mn\epsilon^{-2}= O(mn\epsilon^{-2}).
\end{equation}
In the first analysis  we will prove that when $p \in conv(S)$, the  number of iterations $K_\epsilon$, needed to get an approximate solution $p'$ satisfying (\ref{c1}) is bounded above by $48 \epsilon^{-2}$.   In the worst-case each iteration of Step 1  requires $O(mn)$ arithmetic operations. However, it may also take only $O(m)$ operations.  The number of arithmetic operations in each iteration of Step 2 is only $O(m +n)$
($O(m)$ to find a pivot and the new iterate $p''$, plus $O(n)$ to update the coefficients $\alpha_i'$ of $p''$). Thus  according to this complexity bound the Triangle Algorithm is a fully polynomial-time approximation scheme whose complexity for computing an $\epsilon$-approximate solution is
$O(mn \epsilon^{-2})$ arithmetic operation. In particular, for fixed $\epsilon$ the complexity of the algorithm is only $O(mn)$.

The worst-case iteration complexity estimate stated above is under the assumption of worst-case performance in the reduction of error in each iteration of the algorithm. Also the worst-case arithmetic complexity is under the assumption that in each iteration the algorithm has to go through the entire list of points in $S$ to determine a pivot, or a witness.  Thus our first worst-case arithmetic complexity bound for the Triangle Algorithm is under the assumption that worst-cases happen in each of the two steps.
In practice we would expect a more efficient complexity.   Note that by squaring the distances we have
\begin{equation}
d(p',v_j) \geq d(p,v_j) \iff  \Vert p' \Vert^2 - \Vert p \Vert^2  \geq 2v_j^T(p'-p).
\end{equation}
Thus Step 1 does not require taking square-roots.
Neither does the computation of $p''$. These imply elementary operations are sufficient.

Our second complexity bound takes into account the relative location of $p$ with respect to $S$.  To describe the second complexity bound, let  $A=[v_1, \dots, v_n]$,  the matrix of points in $S= \{v_1, \dots, v_n\}$. We write $conv(A)$ to mean $conv(S)$. Given  $\epsilon \in (0,1)$,  we associate a number $c(p,A, \epsilon)$ to problem (P), called  {\it visibility factor}. It is an indicator of the relative position of $p$ with respect to the iterates in the Triangle Algorithm and points in $S$.

\begin{definition}  Given $p,A$, $p' \in conv(S)$,  and a corresponding pivot $v$, we refer to $\angle pp'v$ as the {\it pivot angle}.
\end{definition}

To describe the visibility factor, let $B_{\epsilon R}(p)$ be the open ball of radius $\epsilon R$ at $p$, where $R= \max\{d(p,v_i), i=1, \dots,n\}$.

\begin{definition} \label{vprime}
Given $p,A, \epsilon$, for each iterate $p' \in conv(A) \setminus B_{\epsilon R}(p)$, let $v_{p'}$ denote the $p$-pivot at $p'$ with the smallest pivot angle $\theta'=\angle pp'v_{p'}$.  Let $\theta_p$ be the maximum of $\theta'$ over all the iterates in the Triangle Algorithm before the algorithm terminates.
\end{definition}

\begin{definition}  Given input $p, A, \epsilon$,  {\it the visibility constant}, $\nu=\nu(p,A,\epsilon)$  of the Triangle Algorithm is $\sin \theta_p$. The {\it visibility factor} is the constant $c=c(p,A,\epsilon)$, satisfying
\begin{equation} \label{nu}
\nu= \frac{1} {\sqrt{1+c}}.
\end{equation}
\end{definition}

As will be shown later, the significance of $\nu$ lies in the fact that if we iterate the Triangle Algorithm $k$ times getting no witnesses, the $k$-th iterate  $p_k \in conv(S)$ satisfies
\begin{equation} \label{nuk}
d(p_k,p) \leq \nu^k d(p_0,p).
\end{equation}
This implies an alternate complexity bound on the number of arithmetic operations of the Triangle Algorithm (our second complexity bound):
\begin{equation} \label{tricomplex}
O\bigg (mn \frac{1}{c}\ln \frac{1}{\epsilon} \bigg).
\end{equation}
Thus combining (\ref{firstbound11}) and (\ref{tricomplex}) we can state the following complexity bound for the Triangle Algorithm as
\begin{equation} \label{tricomplexg}
O\bigg (mn \min \bigg \{ \frac{1}{\epsilon^2}, \frac{1}{c}\ln \frac{1}{\epsilon} \bigg \} \bigg).
\end{equation}
We will show that $c \geq \epsilon^2$. However, depending upon $c$ the Triangle Algorithm could exhibit  very efficient complexity bound, possibly even a polynomial-time complexity bound for solving the convex hull decision problem. Specifically, consider integer input data. If the size of encoding of the convex hull decision problem is $L$ and $1/c$ is a polynomial in $m,n$ and ,$L$, then in polynomial-time complexity  the Triangle Algorithm can compute $p' \in conv(S)$ such that $d(p',p)=d(Ax,p) \leq 2^{-L}$. The representation of such approximate solution $p'$ can be rounded into an exact representation of $p$ as a convex combination of $v_i$'s.

As an example of a case where $c$ can be a very good constant,  we show that when the relative interior of $conv(A)$ contains the ball of radius $\rho$ centered at $p$ (see \cite{Rock} for definition and properties of relative interior),  then $c \geq (\rho/R)^2$. This results in a third complexity bound. Thus when $\rho/R$ is a constant that is independent of $\epsilon$ the Triangle Algorithm performs very well (See Figure \ref{Figzz7}).

\begin{remark}  In case the Triangle Algorithm computes a witness $p' \in conv(S)$ (see (\ref{c2})), we will show $p \not \in conv(S)$ by explicitly describing a hyperplane that separates $p$ from $conv(S)$.  However, if $p \not \in conv(S)$,  for  sufficiently small $\epsilon$ the Triangle Algorithm  will necessarily compute a witness. Since we do not know the magnitude of such $\epsilon$, one possible approach in this case is to first run the algorithm  with a large value of $\epsilon$, say, $\epsilon =0.5$. If it finds a corresponding $\epsilon$-approximate solution instead of a witness, we replace $\epsilon$ with $\epsilon/2$ and repeat this process. The overall complexity would be as if we would set $\epsilon = \Delta/R$, where $\Delta=d(p,conv(S))$, the distance between $p$ and $conv(S)$. The complexity of computing a witness can be derived by substituting $\epsilon = \Delta/R$ in (\ref{tricomplexg})

\begin{equation} \label{tricomplexgg}
O\bigg (mn \min \bigg \{ \frac{R^2}{\Delta^2}, \frac{1}{c}\ln \frac{R}{\Delta} \bigg \} \bigg).
\end{equation}
\end{remark}

When $p \not \in conv(S)$, the Triangle Algorithm does not attempt to compute $p_*$, rather a separating hyperplane.  However, by virtue of the fact that it finds a very special separating hyperplane, i.e. a hyperplane  orthogonally bisecting the line $pp'$, it in the process computes an approximation to $d(p,p_*)=\Delta$ to within a factor of two. More precisely, any witness $p'$ satisfies the inequality
\begin{equation} \label{half}
\frac{1}{2}d(p,p') \leq d(p,p_*)= \Delta \leq d(p,p').
\end{equation}

\begin{remark}
Not only this approximation is useful for problem (P), but for the case of computing the distance between two convex hulls, i.e. the polytope distance problem.
 It is well known that the Minkowski difference of two convex hulls is a polytope whose shortest vector has norm equal to the distance between the two polytopes,  see e.g. Clarkson \cite{clark2008} and G{\"a}rtner and Jaggi \cite{Gartner}. In a forthcoming article, \cite{kal13a}, we will give a Triangle Algorithm that can compute the distance between two convex sets, in particular the case of two compact convex hulls.
\end{remark}

\subsection{The Complexity of Solving LP Via The Triangle Algorithm} \label{subsec2.3}

We consider the applicability of the Triangle Algorithm in order to solve the LP feasibility problem. This may be written as the problem of testing if the polyhedron
\begin{equation} 
\Omega=\{x \in \mathbb{R} ^n:  \quad Ax=b, \quad x \geq 0 \}
\end{equation}
is nonempty, where $A=[a_1, \dots, a_n]$ is an $m \times n$ real matrix. To apply the Triangle Algorithm, either we need to assume a known bound $M$ on the one-norm of the vertices of $\Omega$, i.e. a constraint $\sum_{i=1}^n x_i \leq M$, or it is known $\Omega$ has no {\it recession direction}, i.e. there does not exist a nontrivial $d \geq 0$ with $Ad=0$.  In case of a known bound $M$, by introducing a slack variable, $x_{n+1}$, in the new inequality constraint, as well as augmenting $A$ by a zero column vector,  the LP feasibility problem reduces to the convex hull decision problem of testing if $b/M$ lies in $conv([A,0])=conv(\{a_1, \dots, a_n, 0\})$.  In theory, when no such a bound $M$ is available but the input data are integers, one can argue an upper bound $2^{O(L)}$, where $L$ is the size of encoding of $A$ and $b$.  Such bounds coming from LP-type analysis, initially stated in \cite{kha79} are well-known and discussed in many books. For an analysis, see \cite{kal97}. In practice, one can start with a smaller bound $M$ and gradually increase it while testing feasibility.

If it is known that  $\Omega$ has no  recession direction, then no such bound is necessary. In such case  it is easy to prove $\Omega \not = \emptyset$ if and only if  $0 \in conv([A,-b])=conv(\{a_1, \dots, a_n, -b\})$.  In this case, given any $\epsilon \in (0,1)$,the Triangle Algorithm gives an $\epsilon$-approximate solution:
\begin{equation}
p'=\sum_{i=1}^n \alpha_i a_i - \alpha_{n+1}b \in conv \bigg (\{a_1, \dots, a_n, -b\} \bigg ),
\end{equation}
where
\begin{equation}\label{rprime}
d(p',0)=\Vert p' \Vert < \epsilon R', \quad
R'= \max \bigg \{\Vert a_1 \Vert, \dots, \Vert a_n \Vert, \Vert b \Vert \bigg \}.
\end{equation}
Setting $x_0=(\alpha_1, \dots, \alpha_n)^T/\alpha_{n+1}$,  $x_0 \geq 0$ and from (\ref{rprime}) we have
\begin{equation} \label{accur}
d(Ax_0,b) < \frac{\epsilon R'}{\alpha_{n+1}}.
\end{equation}

A practical and straightforward algorithm for solving LP feasibility via  the Triangle Algorithm is to run the algorithm and  in each iteration check if the bound in (\ref{accur}) is within a prescribe tolerance $\epsilon_0 \in (0,1)$. However, to give a theoretical complexity bound on the number of iterations, we need a lower bound on $\alpha_{n+1}$. We prove a sensitivity theorem (Theorem \ref{thmc}) that provides such a lower bound.  We prove it suffices to have  $\epsilon  \leq {\epsilon_0 \Delta_0}/{4R'}$, where
\begin{equation} \label{delta0}
\Delta_0 = \min \bigg \{\Vert p \Vert:  \quad p \in conv \bigg (\{a_1, \dots, a_n\} \bigg ) \bigg \}=\min \bigg \{
\Vert Ax \Vert :  \quad \sum_{i=1}^nx_i=1, \quad x \geq 0 \bigg \}.
\end{equation}

Indeed we offer a Two-Phase Triangle Algorithm that in Phase I  computes a witness $p' \in conv (A)$ to estimate $\Delta_0$. This follows from the inequality
\begin{equation} \label{halfprime}
\frac{1}{2} \Vert p' \Vert \leq \Delta_0 \leq \Vert p' \Vert.
\end{equation}
It then proceeds to Phase II to compute an approximation to a point $x_0$ so that $
d(Ax_0, b) <  \epsilon_0 R'$.

Table \ref{table1} summarizes the complexity of solving three different LP feasibility problems and LP itself via the Triangle Algorithm.  These will be proved in subsequent sections. The first block corresponds to the convex hull decision problem itself, where specifically $A$ represents the $m \times n$ matrix  of the points, and $b$ represents $p$. Its columns are represented by $a_i$ and $conv(A)$ represents the convex hull of its columns.

The second block correspond to solving $Ax=b, x \geq 0$ when $0 \not \in conv(A)$. The third block corresponds to solving $Ax=b, x \geq 0$ when a bound is given on the sum of the variables.  The forth block corresponds to solving an LP by converting it into a feasibility problem with a known bound.  Specifically, the general LP problem, $\min \{c^Tx : Ax \geq b, x \geq 0\}$ when combined with its dual, $\max \{b^Ty : A^Ty \leq c, y \geq 0\}$, can be formulated as an LP feasibility problem
\begin{equation}
Ax-x_a=b, \quad A^Ty+y_a=c, \quad c^Tx =b^Ty, \quad x, x_a, y, y_a \geq 0.
\end{equation}
This can be written as
$\widehat A \widehat x= \widehat b, \widehat x \geq 0$ where $\widehat A$ is an $O(m+n)$ matrix. Thus, given a bound $\widehat M$ on the feasible solutions,  the complexity to get a solution $\widehat x \geq 0$ such that $d(\widehat A \widehat x, \widehat b) < \epsilon \widehat R$ can be stated.

\begin{table}[htpb]
	\renewcommand{\arraystretch}{1.0}
	\centering
\scalebox{0.95}{
\begin{tabular}{|l|l|l|c|}

\hline
Application  &  Quality of  &  Time Complexity of  & Description of
\\
Type & Approximation & Triangle Algorithm &  Constants
\\
\hline
  & Compute $x \geq 0$ & &
\\
$Ax =b, x \geq 0,$ &  such that & $  O\big (mn \min \{\frac{1}{\epsilon^2}, \frac{1}{c}\ln \frac{1}{\epsilon}\} \big)$ & $R= \max\{d(b,a_i)\}$

\\
$\sum x_i =1$  &  $d(Ax,b) < \epsilon R$ & &  $c=c(b,A, \epsilon) \geq \epsilon^2$
\\ & & &
\\
\hline
 &  Compute  $x \geq 0$ &  & $R'= \max\{ \Vert a_i \Vert , \Vert b \Vert \}$
\\
$Ax=b, x \geq 0$ & such that  &  $O\big (mn \min \{\frac{{R'}^2}{{\epsilon^2 \Delta_0}^2}, \frac{1}{c'}\ln \frac{R'}{\epsilon \Delta_0} \} \big)$& $\Delta_0=d(0, conv(A))$
\\
$0 \not \in conv(A)$ &  $d(Ax,b) < \epsilon R'$  &  & $c'=c(0,[A,-b], \frac{\epsilon \Delta_0}{4R'}) \geq  \frac{{\epsilon^2 \Delta_0}^2} {{4R'}^2}$
\\ & & &
\\
\hline
&  Compute $x \geq 0$ &  &

\\
$Ax=b, x \geq 0$, & such that &  $O\big (mn \min \{\frac{M^2}{\epsilon^2}, \frac{1}{c''}\ln \frac{M}{\epsilon}\} \big)$&$R'= \max\{ \Vert a_i \Vert , \Vert b \Vert \}$
\\
$\sum  x_i \leq M$   & $d(Ax,b) < \epsilon R'$  &  &  $c''=c(\frac{1}{M}b,[A, 0], \epsilon) \geq \frac{\epsilon^2}{M^2}$
\\ & & &
\\
\hline
$\min c^Tx$ &  Compute $\widehat x \geq 0$ &  &

\\
s.t. $A x= b, x \geq 0$, & such that &  $O\big ((m+n)^2 \min \{\frac{\widehat M^2}{\epsilon^2}, \frac{1}{\widehat c}\ln \frac{\widehat M}{\epsilon}\} \big)$&$\widehat R= \max\{ \Vert \widehat a_i \Vert , \Vert \widehat b \Vert \}$
\\

$\widehat A \widehat x= \widehat b, \widehat x \geq 0$  & $d(\widehat A \widehat x, \widehat b) < \epsilon \widehat R$  &  &  $\widehat c=c(\frac{1}{\widehat M} \widehat b,[\widehat A, 0], \frac{\epsilon}{\widehat M}) \geq \frac{\epsilon^2}{\widehat M^2}$
\\$\sum \widehat x_i \leq \widehat M$   & & &
\\
\hline
	\end{tabular}
}
\caption{Complexity of solving via Triangle Algorithm: (1) the convex hull decision problem (first block); (2) LP feasibility with no recession direction (second block); (3) LP feasibility with a known bound on feasible set (third block); (4) LP optimization  converted into an LP feasibility. In all cases $A$ is an $m \times n$ matrix.}
	\label{table1}
\end{table}

\subsection{Triangle Algorithm Versus Simplex Method to Solve Problem (P)} \label{subsec2.4}

The convex hull decision problem is clearly a special LP feasibility, testing the feasibility of $Ax=b, x \geq 0, e^Tx=1$, where $e$ is the $n$-vector of ones.  Without loss of generality we may assume $b_i \geq 0$.  We can formulate the convex hull decision problem as the following LP:
\begin{equation}
\min \bigg \{\sum_{i=1}^{m+1} z_i : Ax + z=b, \quad e^Tx + z_{m+1}= 1, \quad x \geq 0, z_i \geq 0, i=1, \dots, n  \bigg \}.\end{equation}
In the above formulation we have introduced an artificial variable $z_i$  for each of the $m+1$ constraints. The feasible region of the LP is nonempty since we can set $z_i=b_i$, $i=1, \dots, m$ and $z_{m+1}=1$.

By scaling the $x$ component of each basic feasible solution of the LP that arises in the course of applying the simplex method, if $x \not =0$, and $x'=x/e^Tx \geq 0$,
then $p'=Ax'$ lies in $conv(A)$. Thus the simplex method gives rise to a sequence of points in $conv(A)$  the last one of which either proves that $b \in conv(A)$, or if the objective value is nonzero, proves $b \not \in conv(A)$.  If $b \in conv(A)$ and we represent the sequence of distinct iterates of the simplex method by $p_1=Ax^{(1)}, \dots, p_{t-1}=Ax^{(t-1)}, p_t=Ax^{(t)}=b$, then  the objective value is zero only at the last iteration and therefore no iterate of the simplex method can get closer to $b$ than $\delta_s=\min \{d(b, p_i): i=1, \dots, t-1\}$. However, the sequence of iterates in the Triangle Algorithm converge to $b$ and thus after a certain number of iterations the gap will satisfy $\epsilon R \leq \delta_s$.  The precise number of iteration thus can be determined by setting $\epsilon = \delta_s/R$.

The number of elementary operations at each iteration
of the revised simplex method is $O(mn)$. According to Dantzig \cite{Dantzig} and Shamir's survey article \cite{Shamir}, the expected number of iterations of the simplex method to find a feasible solution to a linear program, say $Ax=b, x \geq 0$ where $A$ is $m \times n$, is conjectured to be of the order
$\alpha m$,  $\alpha$ is  $2$ or $3$. Thus based on the above we would expect that the complexity of the revised simplex method to solve the convex hull decision problem to be $O(m^2n)$.
Ignoring the worst-case exponential complexity of the simplex method or possible cycling, and taking the average case complexity for granted when solving the convex hull decision problem, we may ask how does the simplex method compare with the Triangle Algorithm? In $k$ iterations of the Triangle Algorithm, a time complexity of $O(mnk)$, we can computes an approximation $p_k$ so that $d(p_k,p) \leq \nu^k d(p_0,p)$, $\nu$ the visibility constant (see (\ref{nuk})). Depending upon the value of $\nu$ we could obtain an extremely good approximation for $k$ much less than $m$.  In a forthcoming paper we will make a computational comparison between the Triangle Algorithm and the simplex method for solving the convex hull decision problem.

\subsection{Triangle Algorithm Versus Sparse Greedy  to Solve Problem (P)} \label{subsec2.5}

Formally, the distance between $p$ and  $conv(S)$ is defined as
\begin{equation} \label{euclid}
\Delta=d \bigg (p, conv(S) \bigg )= \min \bigg \{d(p',p): \quad p' \in conv(S)\}=d(p_*,p) \bigg \}.
\end{equation}

We have, $p \not \in conv(S)$, if and only if $\Delta >0$.  While solving problem (P) does not require the computation of $\Delta$ when it is positive, in some applications this distance is required. However, as stated in (\ref{half}), any witness approximates $\Delta$ to within a factor of two. This fact indicates another useful property of the Triangle Algorithm.

One of the best known algorithms for determining the distance between two convex polyposes is Gilbert's algorithm,  \cite{Gilbert}.  The connections and equivalence of Gilbert's algorithm and Frank-Wolfe algorithm, a gradient descent algorithm, when applied to the minimization of a convex quadratic over a simplex is formally studied in G{\"a}rtner and Jaggi \cite{Gartner}.  They make use of a notion called {\it coreset}, previously studied in  \cite{clark2008}, and define a notion of $\epsilon$-approximation which is different from our notion given in (\ref{c1}).  Furthermore, from the description of Gilbert's algorithm in \cite{Gartner} it does not follow  that Gilbert's algorithm and the Triangle Algorithm are identical.  However, there are similarities in theoretical performance of the two algorithms and we will discuss these next. Indeed we believe that the simplicity of the Triangle Algorithm and the distance duality theorems that inspires the algorithm, as well as the its theoretical performance makes it distinct from other algorithms for the convex hull decision problem. Furthermore, these features of the Triangle Algorithm may encourage and inspire new applications of the algorithm and further theoretical analysis, in particular amortized complexity of the Triangle Algorithm. In upcoming reports we shall present some such results.

The convex hull decision problem, and its optimization form in (\ref{euclid})
can equivalently be formulated as the minimization of a convex quadratic function over a simplex:
\begin{equation} \label{convexminimization}
\min \bigg \{f(x)=d \bigg (\sum_{i=1}^n x_iv_i, p \bigg )^2 : \quad x \in \Sigma_n \bigg \}, \quad \Sigma_n =\bigg \{x \in \mathbb{R} ^n: \quad \sum_{i=1}^nx_i =1, x \geq 0 \bigg \}.
\end{equation}

Let $x_*$ be an  optimal solution of (\ref{convexminimization}). Suppose  $f(x_*)=0$. Then an approximation algorithm would attempt to compute a point $x'$ in the simplex so that $f(x')$ is small.  Let us examine the performance of Triangle Algorithm. As before, let
\begin{equation}
R= \max \{d(p,v_i): v_i \in S\}.\end{equation}
Given $\epsilon \in (0,1)$, letting $p'=Ax'$, $x' \in \Sigma_n$, then $d(p',p) < \epsilon R$ if and only if $f(x') < \epsilon^2 R^2$.  Thus using Table \ref{table1}, given $\epsilon \in (0,1)$, the complexity of computing $x'$ so that
\begin{equation}
f(x') < \epsilon R^2\end{equation}
can be computed by substituting for $\epsilon$,
$\sqrt{\epsilon}$ in the first complexity block to give
\begin{equation}
O\bigg (mn \min \bigg \{ \frac{1}{\epsilon},  \frac{1}{c} \ln \frac{1}{\epsilon} \bigg \} \bigg ).\end{equation}
Furthermore, if $f(x_*) >0$, and $p'=Ax'$ is a witness, then $0.5 \Delta \leq d(p'p) \leq \Delta$. Since $f(x_*) =\Delta^2$,  we get
\begin{equation}
\frac{1}{4} f(x_*) \leq f(x') \leq f(x_*).\end{equation}

Then by substituting for $\epsilon$ in the first block of Table \ref{table1} the quantity $\Delta/R$, the complexity of computing the above approximation is
\begin{equation}
O\bigg (mn \min \bigg \{\frac{R^2}{f(x_*)},  \frac{1}{c} \ln \frac{R^2}{f(x_*)} \bigg \} \bigg ).\end{equation}

We now contrast this complexity with the {\it greedy algorithm} for the optimization of $f(x)$ (as well as more general smooth convex function $f(x)$), as described in Clarkson \cite{clark2008}. It can equivalently be described as the following concave maximization:

\begin{center}
\begin{tikzpicture}
\node [mybox] (box){%
    \begin{minipage}{0.9\textwidth}
{\bf  Greedy Algorithm}\

\begin{itemize}
\item
{\bf Step I.}
Given $x' \in \Sigma_n$, let $j$ be the index satisfying
 $ \frac{\partial f(x')}{\partial x_j}= \min \{\frac{\partial f(x')}{\partial x_i}, i=1, \dots, n \}.$

\item
{\bf Step II.} Compute  $x''= {\rm argmin} \{f(x'+ \alpha (e_j-x')):  \quad \alpha \in [0,1] \}$, where $e_j$ is the $j$-th  vector of the standard basis. Replace $x'$ with $x''$, go to Step I.
\end{itemize}
    \end{minipage}};
\end{tikzpicture}
\end{center}

Step 1 of the Triangle Algorithm and Step I of the Greedy Algorithm (also known as sparse greedy approximation) have in common the fact that they select an index $j$ so that $v_j$ will be used as a $p$-pivot.  Having computed such a $p$-pivot for $p'$, Step 2 of the Triangle algorithm and Step II of the Greedy Algorithm simply perform a line search.  However, the motivation behind the selection of the index $j$ is very different. The Greedy Algorithm coincides with Frank-Wolfe algorithm and Gilbert's algorithm. The Greedy Algorithm is algebraically motivated (using gradients), while the Triangle Algorithm is geometrically motivated.  The Triangle Algorithm does not need to search over all the indices to find a $p$-pivot $v_j$.  In its best case it finds such $j$ in one iteration over the indices by performing only $O(m)$ arithmetic operations. It then needs to update the representation of the current iterate, taking $O(n)$ operations. Thus one iteration of the Triangle Algorithm could cost $O(m+n)$ operations ($O(m)$ to find a pivot, plus $O(n)$ to update the coefficients $\alpha_i'$ of $p''$).  In the worst-case an iteration will require $O(mn)$ operations. The Greedy Algorithm in contrast requires $O(mn)$ arithmetic operations in every iteration.  This can be seen when $f(x)$ is written as $d(Ax,p)^2$. Its gradient at a point would in particular require the computation of $A^TAx$.

The Greedy Algorithm  generates a sequence of vectors $x_{(k)} \in \mathbb{R} ^n$ where $x_{(k+1)}$ has at most $k$ nonzero coordinates. This is advantageous when $n$ is very large.  As will be easily verifiable this property also holds for the Triangle Algorithm when the initial iterate $p_0$ is taken to be sparse. Another property of the Greedy Algorithm is that if $x_*$ is the optimal solution of (\ref{convexminimization}), then
\begin{equation}
f(x_{(k)})-f(x_*) \leq \frac{C_f}{k}= O \bigg (\frac{1}{k} \bigg ),
\end{equation}
where $C_f$ is a constant that depends on the Hessian of $f$, see
Clarkson \cite{clark2008} and Zhang \cite{zhang}.

The Triangle Algorithm generates a sequence of points $p'_{(k)}$ in $conv(S)$ that get closer and closer to $p$.  This sequence corresponds to a sequence $x'_{(k)}$ in $\Sigma_n$, where $f(x'_{(k)})=d(p'_{(k)}, p)^2$. If $p \in conv(S)$, $f(x_*)=0$.  Given an $\epsilon \in (0,1)$, according to the first complexity bound, the Triangle Algorithm in $K_\epsilon \leq 48\epsilon^{-2}$ iterations will generate a point $p_\epsilon \in conv(S)$ satisfying
\begin{equation}
d(p'_{(K_\epsilon)},p) < \epsilon R, \quad R=\max \big \{d(p,v_i), i=1, \dots, n  \big \}.
\end{equation}
Given an index $k$, by reversing the role of $k$ and $\epsilon$ and solving for $\epsilon$ in the equation  $48/\epsilon^2=k$, we get $\epsilon=\sqrt{48/k}$ so that we may write
\begin{equation}
d(p'_{(k)}, p) < \sqrt{\frac{48}{k}}R.
\end{equation}
Equivalently,
\begin{equation}
f(x'_{(k)}) < \frac{48R^2}{k}=O \bigg (\frac{1}{k} \bigg ).
\end{equation}
In summary, when $p \in conv(S)$ the Triangle Algorithm according one complexity analysis works similar to the Greedy Algorithm, however it may perform better because it only needs to find a pivot $v_j$, as opposed to finding the minimum of partial derivatives in Step I of the Greedy Algorithm. Then when it uses the best pivot strategy, according to the second complexity, it could make much more effective steps than Greedy Algorithm. When $p$ is not in $conv(S)$, the Triangle Algorithm  can use any witness to give an approximation of closest point to within a factor of two, see (\ref{half}). We may conclude that the Triangle Algorithm in theory is at least as effective as the Greedy Algorithm, and possibly faster whether approximating $p \in conv(S)$, or estimating the distance $\Delta$ to within a factor of two.

In the context support vector machines (SVM) (see \cite{Burg} for applications),  Har-Peled et al. \cite{Har} use coreset to give an approximation algorithm, see also Zimak \cite{Zimak}. We mention these because we feel that the Triangle Algorithm, despite some similarities with existing algorithms or their analysis, is distinct from them. It is quite simple and geometrically inspired by the distance duality, a simple but new and rather surprising property.

In fact the complexity of the Triangle Algorithm could be much more favorable in contrast with these algorithms.  As mentioned earlier the Triangle Algorithm is not designed to approximate $f(x_*)$ to prescribed tolerance when  $f(x_*) >0$.  However, in a forthcoming article \cite{kal13a}, we describe a generalization of the Triangle Algorithm that in particular approximates the minimum value $f(x_*)$ to prescribed tolerance, almost with similar complexity.

\subsection{Triangle Algorithm Versus First-Order Methods to Solve Problem (P)} \label{subsec2.6}

Consider the problem
\begin{equation} \label{firstorder}
\min  \bigg \{ g(x) = d(A x, b)=\sqrt{f(x)}, x \in \Sigma_n  \bigg \}, \quad
\Sigma_n =\bigg \{x \in \mathbb{R} ^n: \quad \sum_{i=1}^n x_i =1, x_i \geq 0, i=1, \dots, n \bigg \}.
\end{equation}
When $p \in conv(S)$, to compute $x_\epsilon$ such that $g(x_\epsilon) < \epsilon$
can be solved in   $O(\sqrt{\ln n} \max \| a_i \|_2 /\epsilon)$ iterations of the fast gradient scheme of Nesterov \cite{Nesterov}, so-called $O(1/\epsilon)$-method,
as applied to a smoothed  version of $g(x)$, where each iteration takes $O(mn)$ operations,  \cite{Nesterov2}.  A cone programming version of such method is described in Lan et al \cite {Lan}.

Let us compare this complexity with the Triangle Algorithm.  In one analysis of the complexity,  the Triangle Algorithm computes an approximate solution $x_\epsilon \in \Sigma_n$ where $g(x_\epsilon) < \epsilon R$ in $O(1/\epsilon^2)$ iterations. Each iteration takes  $mn$ operations in the worst case when any pivot is used, but an iteration could only take $O(m+n)$ arithmetic operations.  Changing $\epsilon$ in Nesterov's method to $\epsilon R$, we conclude Nesterov's method computes $x'_\epsilon$ such that $g(x'_\epsilon) < \epsilon R$ in
$O(\sqrt{\ln n}/\epsilon)$ iterations as opposed to $1/\epsilon^2$.  However, as we shall show when $p \in conv(S)$ and its distance to a boundary point of $S$ is at least $\epsilon^{1/t} R$, $t$ a natural number, the number of iterations of the Triangle Algorithm is $O(\sqrt[t]{\epsilon^{-2}}\ln \delta_0/\epsilon R)$, $\delta_0=d(p_0,p)$, initial gap. This already makes it comparative or better than the first-order methods for testing if $p \in conv(S)$.

\subsection{Problem (P) and Linear Programming Algorithms} \label{subsec2.7}

Linear programming has found numerous practical and theoretical applications in applied mathematics, computer science, operation research and more. In particular, the simplex method of Dantzig is not only a significant algorithm for solving LP but also a theoretical tool to prove many results.  Ever since the  Klee-Minty \cite{Klee} example showed exponential worst-case time complexity of the simplex method, many LP algorithms have been invented. The trend will most likely continue.

Problem (P) is a very special case of the LP feasibility problem. However, in fact the general LP with integer inputs can be formulated as a homogeneous case of problem (P), i.e. $p=0$.  The corresponding problem (P), may be referred as {\it homogeneous feasibility problem} (HFP), see  \cite{kal88}.  A classical  duality corresponding to HFP is Gordan's theorem, a special case of the separating hyperplane theorem, easily provable from Farkas lemma:  either $0 \in conv(S)$, or there exists $y \in \mathbb{R} ^m$ such that $y^Tv_i >0$ for all $i=1, \dots, n$, see e.g. Chv\'atal \cite {chvatal2}. It can be justified that both Khachiyan  and  Karmarkar algorithms  explicitly or implicitly are designed to solve HFP. This is because on the one hand Karmarkar's  canonical formulation of LP can easily be converted into an HFP. On the other hand,
Khachiyan's ellipsoid algorithm solves a system of strict inequalities ($Ax < b$) whose alternative system, by Gordan's theorem is an HFP ($A^Ty=0, b^Ty+s=0, \sum y_i+s=1, y \geq 0, s \geq 0)$). For this and additional results on the connections between HFP and LP feasibility,  see \cite{jinkal}.

By exploring the close relationship between HFP, equivalent to the problem of computing a nontrivial nonnegative zero of a quadratic form, and the diagonal matrix scaling  problem,   Khachiyan and Kalantari \cite{kha90} have given a very simple path-following polynomial-time algorithm for LP as well as for quasi doubly stochastic diagonal scaling of a positive semidefinite matrix. In particular, the algorithm can test the existence of an $\epsilon$-approximate solution of HFP in a number of arithmetic operations proportional to $n^{3.5}$ and $\ln \epsilon^{-1}$.  As approximation schemes, all known polynomial-time algorithms for LP have a complexity that is polynomial in the dimension of the data and in $\ln \epsilon^{-1}$.  As is well known, an exact solution for an LP with integral input can be computed by rounding any approximate solution having sufficient precision. Even if the complexity of a polynomial-time algorithms for LP would allow solving problem (P) to within $\epsilon$ accuracy in $O(m^2n \ln\epsilon^{-1})$ arithmetic operations, the Triangle Algorithm still offers an attractive alternative when the dimensions of the problems are large, or when the visibility constant is good.

Other algorithms for LP include, Megiddo's algorithm \cite{meg84} with a running time that for fixed $m$  is linear in $n$, however, has exponential complexity in $m$. Since Mediggo's work a number of randomized LP algorithms have been devised, e.g. Dyer and Frieze \cite{Dyer}, Clarkson \cite{Clar88}, Seidel \cite{Sei}, Sharir and Welzel \cite{Shar}.  Kalai \cite{Kalai} gave a randomized LP simplex method with subexponential complexity bound.  Matou{\v s}ek, Sharir, and Welzl \cite{Mat92} proved another randomized subexponential complexity algorithm for LP. See also  Motawani and Raghavan \cite{Mota}. Kelner and Spielman \cite{ks} have given the first randomized polynomial-time simplex method that analogous to the other known polynomial-time algorithms for linear programming has a running time dependent polynomially on the bit-length of the input.  A  history of linear programming algorithms from computational, geometric, and complexity  points of view that includes simplex, ellipsoid, and interior-point methods is given in Todd \cite{todd}.

\subsection{Outline of Results}  \label{subsec2.8}
The remaining sections of the article are as follows.  In Section \ref{sec3}, we prove several characterization theorems leading into the {\it distance duality} theorem. We then describe several associated geometric properties and problems, as well as  generalizations of the distance duality. In Section \ref{sec4}, while using purely geometric arguments, we give an analysis of the worst-case reduction of the gap in moving from one approximation in the convex hull of points to the next. In Section \ref{sec5}, we formally describe the steps of the Triangle Algorithm. We then use the results in Sections \ref{sec3} and \ref{sec4} to derive a bound on the worst-case complexity of the Triangle Algorithm. In Section \ref{sec6}, we
prove a {\it strict distance duality} and define corresponding strict pivots and strict witness. We also prove a minimax theorem related to the strict distance duality.  In Section \ref{sec7}, we derive an alternate complexity bound on the Triangle Algorithm in terms of constants defined as {\it visibility constant} and {\it visibility factor}. We analyze this alternate complexity in
terms  of these constants and their relations to the location of $p$ relative to $S$. These suggest that the Triangle Algorithm could result in a very efficient algorithm.  In Section  \ref{sec8}, we describe {\it Virtual Triangle Algorithm} and its approximated version. These serve are fast algorithms that if successful can efficiently lead to a proof of infeasibility, or to the computation of a good approximate solution. In Section \ref{sec9}, we describe {\it auxiliary pivots}, points that can be added to $S$ so as to improve visibility constants, hence the computational efficiency  of the Triangle Algorithm. In Section \ref{sec10}, we describe {\it Generalized Triangle Algorithm} based on a {\it generalized distance duality}.  Each iteration of the Generalized Triangle Algorithm computes a more effective approximation, however at the cost of more computation. In Section \ref{sec11}, we consider solving the LP
feasibility problem having no recession direction.  We prove a {\it sensitivity theorem} that gives the necessary accuracy in solving a corresponding convex hull decision problem. Using the sensitivity theorem we  derive  complexity bound for computing an approximate feasible point, via the straightforward application of the Triangle Algorithm, or a  {\it Two-Phase} Triangle Algorithm. In Section \ref{sec12}, we derive the complexity for solving the general LP feasibility via the Triangle Algorithm making no assumption on the recession direction, however with a known bound  on the one-norm of the vertices. In particular, this gives a complexity bound for solving a general LP optimization via the Triangle Algorithm. Finally, we conclude the article with some remarks on applications and extensions of the results.


\section{Characterizations and Applications}  \label{sec3}

Throughout the section, let $S=\{v_1, \dots, v_n\} \subset \mathbb{R} ^m$, and $p$ a distinguished point in $\mathbb{R} ^m$.

\begin{thm} \label{thm1} {\bf (Characterization of Feasibility) } $p \in conv(S)$ if and only if given any  $p' \in conv(S) \setminus \{p\}$, there exists $v_j \in S$ such that $d(p',v_j) > d(p,v_j)$.
\end{thm}

\begin{proof} Suppose $p \in conv(S)$. Consider the Voronoi cell of $p$ with respect to the two point set $\{p, p'\}$, i.e. $V(p)= \{x \in \mathbb{R} ^m:  \quad d(x,p) < d(x,p')\}$ (see Figure \ref{Fig1}).  We claim there exists $v_j \in V(p)$.  If not, $S$ is a subset of  $\overline V(p')= \{x \in \mathbb{R} ^m:  \quad d(x,p) \leq d(x,p')\}$.  But since $\overline V(p')$ is convex it contains $conv(S)$.
But since $V(p) \cap \overline V(p')= \emptyset$ this contradicts that $p \in conv(S)$.

Conversely, suppose that for any $p' \in conv(S) \setminus \{p \}$ there exists $v_j$ such that $d(p',v_j) > d(p,v_j)$. If $p \not \in conv(S)$, let $p' \in conv(S)$ be the closest point to $p$.  Since the closest point is unique, for each $i=1, \dots, n$, in the triangle $\triangle pp'v_i$ the angle $\angle pp'v_i$ must be non-acute. Hence, $d(p',v_i) < d(p,v_i)$ for all $i$, a contradiction. Thus, $p \in conv(S)$.
\end{proof}

\begin{remark}  \label{rm0}  We may view Theorem \ref{thm1} as a characterization theorem for feasibility or infeasibility of a point with respect to the convex hull of a finite set of points.  The present proof is a simpler version of a proof given in the first version of the present article, \cite{kal12}. It was brought to our attention that a proof by Kuhn \cite{kuhn} was given for points in the Euclidean plane. Kuhn's  proof makes use of several results, including Ville's Lemma.  Kuhn's proof makes no connections to Voronoi diagram which is important in the development of the Triangle Algorithm and its analysis.  Some generalizations of the theorem over normed spaces is given by Durier and Michelot \cite{DM86}.  In this article we offer other generalizations and use them in variations of the Triangle Algorithm.  We refer to Theorem \ref{thm1} as {\it distance duality} because it is based on comparisons of distances.  It can be viewed as a
stronger version of the classical Gordan's Theorem on separation of zero from the convex hull of finite set of points. Gordan's Theorem and its conic version, Farkas Lemma, are theorems of the alternative and closely related to duality theory in linear programming. In a forthcoming article, \cite{kal13a},  we prove a substantially general version of Theorem \ref{thm1} that gives rise to an algorithm for approximation of the distance between two compact convex subsets of the Euclidean, and a separating hyperplane if they ate disjoint.
\end{remark}

The following is a convenient restatement of Theorem \ref{thm1}, relaxing the strict inequality, $d(p',v_j) > d(p,v_j)$. It has an identical proof to that theorem.

\begin{figure}[htpb]
	\centering
	
	\begin{tikzpicture}[scale=0.4]
			
\draw (0.0,0.0) -- (4,3.0) -- (8,2.0) --(7,0)--(5,-2)-- cycle;
		\draw (0,0) node[below] {$v_1$};
		\draw (7,0) node[right] {$v_2$};
		\draw (4,3) node[above] {$v_3$};
		\draw (5,-2) node[below] {$v_4$};
		\draw (8,2) node[above] {$v_5$};
\filldraw (5,-2) circle (2pt);
\filldraw (0,0) circle (2pt);
\filldraw (8,2) circle (2pt);
\filldraw (7,0) circle (2pt);
\filldraw (4,3) circle (2pt);
\filldraw (1,5) circle (2pt) node[above] {$p$};
\filldraw (3,1) circle (2pt) node[below] {$p'$};
\filldraw (2,3) circle (2pt) node[below] {$\mu~$};
\begin{scope}[red]
\draw (-2,1) -- (8,6);
\end{scope}
\draw (8,6) node[above] {$l$};
\draw (1,5) -- (3,1);
	

\draw (15.0,0.0) -- (19,3.0) -- (23,2.0) --(22,0)--(20,-2)-- cycle;
		\draw (15,0) node[below] {$v_1$};
		\draw (22,0) node[right] {$v_2$};
		\draw (19,3) node[above] {$v_3$};
		\draw (20,-2) node[below] {$v_4$};
		\draw (23,2) node[above] {$v_5$};
\filldraw (20,-2) circle (2pt);
\filldraw (15,0) circle (2pt);
\filldraw (22,0) circle (2pt);
\filldraw (19,3) circle (2pt);
\filldraw (23,2) circle (2pt);
\filldraw (22,2) circle (2pt) node[above] {$p$};
\filldraw (18,1) circle (2pt) node[below] {$p'$};
\filldraw (20,1.5) circle (2pt) node[below] {$\mu~$};
\draw (18,1) -- (22,2);
\begin{scope}[red]
\draw (19,5.5) -- (21,-2.5);
\end{scope}
\draw (19,5.5) node[right] {$l$};
	\end{tikzpicture}
	
	\caption{Example of cases where orthogonal bisector of $pp'$ does and does not separate $p$ from $conv(S)$.}
	\label{Fig1}
\end{figure}
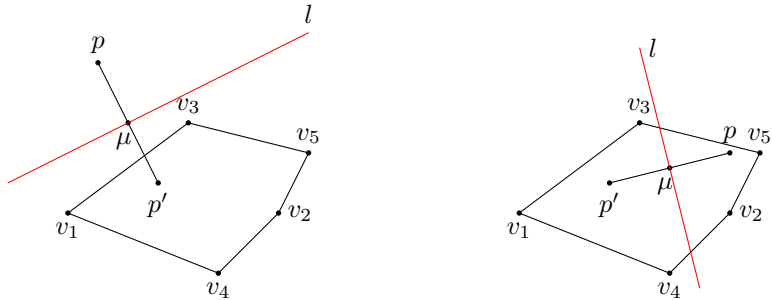

\begin{thm} \label{thm1eq} $p \in conv(S)$ if and only if given any  $p' \in conv(S) \setminus \{p\}$, there exists $v_j \in S$ such that $d(p',v_j) \geq d(p,v_j)$. \qed
\end{thm}

\begin{definition} \label{defn4} We call a point $p' \in conv(S)$ a $p$-{\it witness} (or simply a {\it witness}) if $d(p',v_i) < d(p,v_i)$, for all $i=1, \dots, n$. We denote the set of all $p$-witnesses by $W_p$.
\end{definition}

The following is a characterization of infeasibility.

\begin{thm}  \label{thm2}   {\bf (Characterization of Infeasibility)} $p \not \in conv(S)$ if and only if there exists a $p$-witness $p'$.
\end{thm}

\begin{proof} Suppose $p \not \in conv(S)$. Then by Theorem \ref{thm1eq} there exists $p' \in conv(S)$ such that
$d(p',v_i) < d(p,v_i)$, for all $i=1, \dots, n$. But then $p'$ is a $p$-witness.  Conversely, given that $p'$ is a $p$-witness,  Theorem \ref{thm1} implies $p \not \in conv(S)$.
\end{proof}

Thus we may conclude the following, a non-traditional duality for the convex hull decision problem.

\begin{thm}  \label{thm3duality} {\bf (Distance Duality)}  Precisely one of the two conditions is satisfied:\

(i) For each  $p' \in conv(S)$ there exists $v_j \in S$ such that $d(p',v_j) \geq d(p,v_j)$.

(ii) There exists a $p$-witness. \qed
\end{thm}

The following is a straightforward but geometrically appealing characterization of the set of $p$-witnesses as the intersection of open balls and $conv(S)$.

\begin{prop} \label{prop1} Let $B_i= \{x \in \mathbb{R}^m:  \quad d(x,v_i) < d(p,v_i) \}$, $i=1, \dots, n$.  Then
$W_p= conv(S) \cap (\cap_{i=1}^n B_i)$.  In particular, $W_p$ is a convex subset of  $conv(S)$. \qed
\end{prop}

Figure \ref{Fig2} gives a case when $W_p$ is empty. Figure \ref{Fig3}  gives several scenarios when $W_p$ is nonempty. Next, we state a corollary of Theorem \ref{thm1} and Theorem \ref{thm1eq}.

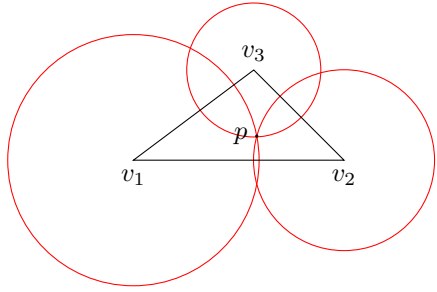
\begin{figure}[htpb]
	\centering
	\begin{tikzpicture}[scale=0.4]
		
		\begin{scope}[red]
		\draw (0.0,0.0) circle (4.177319714841085);
		\draw (7.0,0.0) circle (3.008321791298265);
		\draw (4.0,3.0) circle (2.2278820596099706);
		\end{scope}
		
		\draw (0.0,0.0) -- (7.0,0.0) -- (4.0,3.0) -- cycle;
		\draw (0,0) node[below] {$v_1$};
		\draw (7,0) node[below] {$v_2$};
		\draw (4,3) node[above] {$v_3$};
		
		\filldraw (4.1,0.8) circle (1pt) node[left] {$p$};
	\end{tikzpicture}
	
\begin{center}
\caption{A case with no $p$-witnesses: $p  \in conv(S)$.} \label{Fig2}
\end{center}
\end{figure}

\begin{cor}  \label{cor1}  {\bf (Intersecting Balls Property)}   Consider the set of open balls $B_i= \{x \in \mathbb{R} ^m:  \quad d(x,v_i) < r_i \}$. Assume they have a common boundary point $p$, i.e. $p \in \cap_{i=1}^n \partial B_i= \{x \in \mathbb{R}^m : \quad d(x,v_i) = r_i \}$. Let $S=\{v_1, \dots, v_n\}$.
Then  $p \in conv(S)$ if and only if $(\cap_{i=1}^n B_i) \cap conv(S)= \emptyset$ {\rm (see Figure \ref{Fig2})}.
\end{cor}

\begin{proof}
Suppose $p \in conv(S)$. Pick any point $p' \in conv(S) \setminus \{p\}$. Then by Theorem \ref{thm1} there exists $v_j$ such that $d(p', v_j) > d(p,v_j)$.  But this implies $p' \not \in B_j$, hence $(\cap_{i=1}^n B_i) \cap conv(S)= \emptyset$.

Conversely, suppose that  $(\cap_{i=1}^n B_i) \cap conv(S) = \emptyset$. Thus for each
$p' \in  conv(S)$ there exists $v_j$ such that $d(p', v_j) \geq d(p, v_j)$.  Then by Theorem \ref{thm1eq} we have $p \in conv(S)$.
\end{proof}

\begin{prop}  \label{propnew} {\bf (The Orthogonal Bisector Property)}  Suppose $p'$ is a $p$-witness, i.e. $p' \in conv(S)$ satisfies $d(p', v_i) < d(p,v_i)$, for all $i=1, \dots, n$. Then, the orthogonal bisector hyperplane of the line segment $pp'$ separates $p$ from $conv(S)$. More specifically, let $c=p-p'$ and $\gamma = \frac{1}{2} (\Vert p \Vert^2- \Vert p' \Vert^2)$. If
$H=\{x \in \mathbb{R} ^m:  \quad c^Tx = \gamma \}$, then
$p \in H_{+}=\{x \in \mathbb{R} ^m:  \quad c^Tx >  \gamma  \}$ and  $conv(S) \subset   H_{-}=\{x \in \mathbb{R} ^m:  \quad c^Tx < \gamma \}$.
\end{prop}

\begin{proof}  It is easy to verify that $p \in \overline H_+$.  We claim $v_i \in H_{-}$, for each $i$.  For  each $i=1, \dots, n$ we have, $d^2(p',v_i) < d^2(p,v_i)$. Equivalently, $(p'-v_i)^T(p'-v_i) <  (p-v_i)^T(p-v_i)$. Simplifying, gives
\begin{equation}
2(p-p')^Tv_i < (\Vert p \Vert^2- \Vert p' \Vert^2).\end{equation}

Hence $S \subset H_{-}$. Since $H_{-}$ is convex, any convex combination of points in $S$ is also in $H_{-}$.
\end{proof}

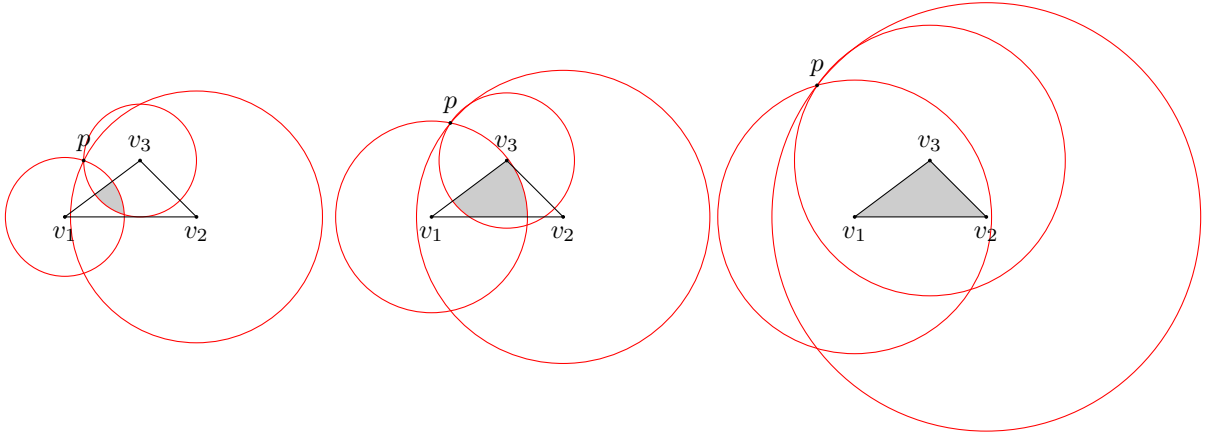
\begin{figure}[htpb]
	\centering
	
	\begin{tikzpicture}[scale=0.25]
		\begin{scope}[red]
		 \clip  (0.0,0.0) circle (3.162);
		 \clip (7.0,0.0) circle (6.7);
		 \clip (4.0,3.0) circle (3);
          \clip (0.0,0.0) -- (7.0,0.0) -- (4.0,3.0) -- cycle;
\fill[color=gray!39] (-10, 10) rectangle (10, -10);
\end{scope}

\begin{scope}[red]
         \draw (0.0,0.0) circle (3.162);
		 \draw (7.0,0.0) circle (6.7);
		 \draw (4.0,3.0) circle (3);
\end{scope}

\draw (0.0,0.0) -- (7.0,0.0) -- (4.0,3.0) -- cycle;
		\draw (0,0) node[below] {$v_1$};
		\draw (7,0) node[below] {$v_2$};
		\draw (4,3) node[above] {$v_3$};
\filldraw (0,0) circle (2pt);
\filldraw (7,0) circle (2pt);
\filldraw (4,3) circle (2pt);
\filldraw (1,3) circle (2pt) node[above] {$p$};
\begin{scope}[red]
		 \clip  (19.5,0.0) circle (5.1);
		 \clip (26.5,0.0) circle (7.8);
		 \clip (23.5,3.0) circle (3.6);
          \clip (19.5,0.0) -- (26.5,0.0) -- (23.5,3.0) -- cycle;
\fill[color=gray!39] (-30, 10) rectangle (30, -10);
\end{scope}

\begin{scope}[red]
         \draw (19.5,0.0) circle (5.1);
		 \draw (26.5,0.0) circle (7.8);
		 \draw (23.5,3.0) circle (3.6);
\end{scope}
			
\draw (19.5,0.0) -- (26.5,0.0) -- (23.5,3.0) -- cycle;
		\draw (19.5,0) node[below] {$v_1$};
		\draw (26.5,0) node[below] {$v_2$};
		\draw (23.5,3) node[above] {$v_3$};
\filldraw (19.5,0) circle (2pt);
\filldraw (26.5,0) circle (2pt);
\filldraw (23.5,3) circle (2pt);
\filldraw (20.5,5) circle (2pt) node[above] {$p$};
\begin{scope}[red]
		 \clip  (42.0,0.0) circle (7.28);
		 \clip (49.0,0.0) circle (11.4);
		 \clip (46.0,3.0) circle (7.2);
          \clip (42.0,0.0) -- (49.0,0.0) -- (46.0,3.0) -- cycle;
\fill[color=gray!39] (10, 10) rectangle (90, -10);
\end{scope}

\begin{scope}[red]
         \draw (42.0,0.0) circle (7.28);
		 \draw (49.0,0.0) circle (11.4);
		 \draw (46.0,3.0) circle (7.2);
\end{scope}
		
		\draw (42.0,0.0) -- (49.0,0.0) -- (46.0,3.0) -- cycle;
		\draw (42,0) node[below] {$v_1$};
		\draw (49,0) node[below] {$v_2$};
		\draw (46,3) node[above] {$v_3$};
\filldraw (42,0) circle (2pt);
\filldraw (49,0) circle (2pt);
\filldraw (46,3) circle (2pt);
		
		\filldraw (40,7) circle (2pt) node[above] {$p$};

	\end{tikzpicture}
	\caption{Examples of nonempty $p$-witness set $W_p$, gray areas: $p \not \in conv(S)$.}
\label{Fig3}
\end{figure}

We now give a complete characterization of the $p$-witness set.

\begin{thm}  \label{thmWit} {\bf (Characterization of $p$-Witness Set)}
$p' \in W_p$ if and only if the orthogonal bisector  hyperplane of the line segment $pp'$ separates $p$ from $conv(S)$.
\end{thm}

\begin{proof}  By Proposition \ref{propnew}, $p' \in W_p$ implies the orthogonal bisector  hyperplane of the line segment $pp'$ separates $p$ from $conv(S)$. Conversely, suppose for some $p' \in conv(S)$
the orthogonal bisector  hyperplane of the line segment connecting  $pp'$ separates $p$ from $conv(S)$. Then, in particular we have $d(p', v_i) < d(p, v_i)$, for all $i=1, \dots, n$.  Hence, $p' \in W_p$.
\end{proof}

\begin{cor}  \label{cornew} Suppose $p \not \in conv(S)$.  Let
\begin{equation}
\Delta=d(p, conv(S))= \min \{d(p,x): \quad x \in conv(S)\}.\end{equation}
Then any $p' \in W_p$ gives an estimate of $\Delta$ to within a factor of two. More precisely,
\begin{equation}
 \frac{1}{2} d(p,p') \leq  \Delta \leq d(p,p').\end{equation}
\end{cor}

\begin{proof} The inequality $\Delta \leq d(p,p')$ is obvious. The first inequality follows from Theorem \ref{thmWit}.
\end{proof}

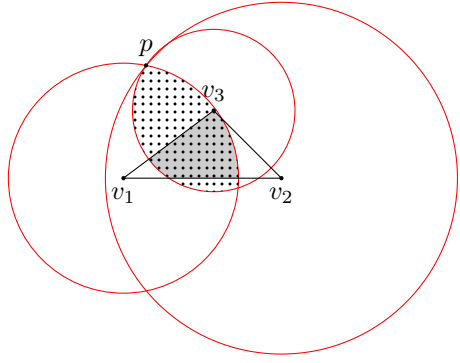
\begin{figure}[htpb]
	\centering
	
	\begin{tikzpicture}[scale=0.3]
		
\begin{scope}[red]
		 \clip  (0,0.0) circle (5.1);
		 \clip (7,0.0) circle (7.8);
		 \clip (4,3.0) circle (3.6);
          \clip (0,0.0) -- (7,0.0) -- (4,3.0) -- cycle;
\fill[color=gray!39] (-10, 10) rectangle (10, -10);
\end{scope}

\begin{scope}[red]
         \draw (0,0.0) circle (5.1);
		 \draw (7,0.0) circle (7.8);
		 \draw (4,3.0) circle (3.6);
\end{scope}
			
\draw (0,0.0) -- (7,0.0) -- (4,3.0) -- cycle;
		\draw (0,0) node[below] {$v_1$};
		\draw (7,0) node[below] {$v_2$};
		\draw (4,3) node[above] {$v_3$};
\filldraw (0,0) circle (2pt);
\filldraw (7,0) circle (2pt);
\filldraw (4,3) circle (2pt);
\filldraw (1,5) circle (2pt) node[above] {$p$};
\def\firstcircle{(0.0,0.0) circle (5.1)}
\def\secondcircle{(7.0,0.0) circle (7.8)}
\def\thirdcircle{(4.0,3.0) circle (3.6)}
    \begin{scope}
      \clip \firstcircle;
      \clip \secondcircle;
      \fill[color=red, pattern=dots]  \thirdcircle;
    \end{scope}

	\end{tikzpicture}
	\caption{The set $\overline W_p$ of  general $p$-witness (dotted area) is a superset of $W_p$ (gray area).}
\label{Fig4}
\end{figure}

Next we give another  characterization of a witness in terms of forbidden zones.

\begin{definition} \label{defn-fz}
The \emph{forbidden zone} $F(X,q)$ for a region $X  \subseteq {\mathbb R}^{m}$ and a point $q \in X$ is the set of all points that
are closer to some point $y \in X$ than $y$ is to $q$, i.e.

\begin{equation}
F(X,q) = \bigg \{z: d(z,y) < d(y,q) \text{ for some } y \in X \bigg \}.
\end{equation}
\end{definition}

The following is simple to prove.

\begin{thm} {\rm (\cite{DKK2})} \label{thm-fz-balls} For a polytope $P$ with vertices $w_i$, $i = 1, \ldots, n$, the forbidden zone $F(P,q)$ is the union of the open balls centered at each $w_i$ with radius $d(q,v_i)$. That is:
\begin{equation}
F(P,q) = \bigcup_{i = 1}^n \bigg \{ z : d(z,w_i) < d (z,q) \bigg \}. \qed
\end{equation}
\end{thm}

We may thus conclude:

\begin{cor} Given $S=\{v_1, \dots, v_n\}$ and $p$, $p' \in W_p$ if and only if $F(conv(S), p')$ does not contain $p$. \qed
\end{cor}

Visually speaking, this says when $p \not \in conv(S)$ and $p'$ is a witness, the union of the open balls centered at $v_i$, having radius $d(p',v_i)$, a region that contains $conv(S)$,  excludes $p$.

Before we utilize the characterization theorems proved here we wish to give a variation of Theorem \ref{thm1}. The theorem shows that the notion of a $p$-witness need not be restricted to the convex hull of $S$. The proof is identical to the proof of Theorem \ref{thm1} and is omitted.

\begin{thm} \label{thm1g} $p \in conv(S)$ if and only if for any point $p' \not = p$, there exists $v_j \in S$ such that $d(p',v_j) > d(p,v_j)$. \qed
\end{thm}

We may thus give a more general distance duality as well as definition for a $p$-witness.

\begin{definition}  We call a point $p' \in \mathbb{R}^m$ a {\it general $p$-witness} if $d(p',v_i) < d(p,v_i)$, for all $i=1, \dots, n$. We denote the set of all general $p$-witnesses by $\overline W_p$.
\end{definition}

Figure \ref{Fig4} depicts the set $\overline W_p$ which of course contains $W_p$ as a subset.  The following is a corollary of Theorem \ref{thm1eq} and Theorem \ref{thm1g}

\begin{cor}  \label{cor1g}  {\bf (General Intersecting Balls Property)} Consider the set of open balls $B_i= \{x \in \mathbb{R} ^m:  \quad d(x,v_i) < r_i \}$. Assume $p \in \cap_{i=1}^n \partial B_i$. Then  $p \in conv(\{v_1, \dots, v_n\})$ if and only if $(\cap_{i=1}^n B_i) = \emptyset$.
\end{cor}

\begin{proof}
Suppose $p \in conv(S)$. Pick any point $p' \not =p$. Then by Theorem \ref{thm1g} there exists $v_j$ such that $d(p', v_j) > d(p, v_j)$.  But this implies $p' \not \in B_i$, hence $(\cap_{i=1}^n B_i)= \emptyset$.

Conversely, suppose that  $(\cap_{i=1}^n B_i) = \emptyset$. In particular, for each
$p' \in  conv(S)$ there exists $v_j$ such that $d(p', v_j) \geq d(p, v_j)$.  Then by Theorem \ref{thm1eq} we have $p \in conv(S)$.
\end{proof}

\begin{remark}  The general open balls property suggests that in proving the infeasibility of $p$ we have the freedom of choosing a $p$-witnesses outside of the convex hulls of $S$.  However, algorithmically it may have no advantage over the Triangle Algorithm to be formally described later.
\end{remark}

\section{Reduction of Gap and Its Worst-Case Analysis}  \label{sec4}

In what follows we state a theorem that is fundamental in the analysis of the algorithm to be described in the next section. It relates to one iteration of the algorithm to test if $p$ lies in $conv(S)$ and reveals its worst-case performance. Before formally analyzing the worst-case of the Triangle Algorithm we give a definition.

\begin{definition} Given $p' \in conv(S)$, we say $v_j \in S$ is a {\it pivot}
relative to $p$ (or $p$-pivot) if $d(p',v_j) \geq d(p,v_j)$.
\end{definition}

In the theorem below the reader may consider $p'$ as a given point in $conv(S)$ and $v$ as a point in $S$  to be used as a {\it $p$-pivot} in order to compute a new point $p''$ in $conv(S)$ where the new gap $d(p'',p)$  is to be a reduction of the current  gap $d(p',p)$.


\begin{thm}  \label{thm3}  Let $p, p', v$  be distinct points in $\mathbb{R} ^m$. Suppose $d(p,v) \leq d(p',v)$.  Let $p''$ be the point on the line segment $p'v$  that is closest to $p$.  Assume $p'' \not =v$. Let $\delta=d(p',p)$, $\delta'=d(p'',p)$, and $r=d(p,v)$ {\rm (see Figure \ref{Fig5})}. Then,

\begin{equation} \label{gap}
\delta' \leq
\begin{cases}
\delta \sqrt{1- \frac{\delta^2}{4 r^2}}, &\text{if $\delta \leq r$;}\\
r, &\text{otherwise.}
\end{cases}
\end{equation}
\end{thm}

\begin{proof} Given $\delta \leq r$, consider $p'$ as a variable $x'$ and the corresponding $p''$ as $x''$. We will consider the maximum value of $d(x'',p)$ subject to the desired constraints. We will prove
\begin{equation}
\delta^* = \max \bigg \{d(x'',p):  \quad x \in \mathbb{R} ^m, \quad d(x',p) =\delta, \quad d(x',v) \geq r \bigg \}= \delta \sqrt{1- \frac{\delta^2}{4 r^2}}.
\end{equation}
This optimization problem can be stated in the two-dimensional Euclidean plane.  Assume $p \not = p''$, and consider the two-dimensional plane that passes through the points $p,p',v$.  Given that $\delta \leq r$, $p'$ must lie inside or on the boundary of the circle of radius $\delta$ centered at $p$, but outside or on the boundary of the circle of radius $r$ centered at $v$, see Figure \ref{Fig5}, circles $C$, $C'$, and $C''$.

Now consider the circle of radius $\delta$ centered at $p$,  $C''$ in Figure \ref{Fig5}.
Consider the ratio $\delta'/r$ as $p'$ ranges over all the points on the circumference of $C''$ while
outside or on the boundary of $C'$. It is  geometrically obvious and easy to argue that this ratio is maximized when $p'$ is a point of intersection of the circles $C'$ and $C''$, denoted by $x^*$ in Figure \ref{Fig6}. We now compute the corresponding ratio.

Considering Figure \ref{Fig6}, and the isosceles triangle  $\triangle vpx^*$, let $h$ denote the length of the bisector line from $v$ to the base, and let $q$ denote the midpoint of $p$ and $x*$.  Consider the right triangles $\triangle pvq$  and  $\triangle px^*x^{**}$. The angle  $\angle vpq$ is identical with  $\angle px^*x^{**}$. Hence, the two triangles are similar and we may write

\begin{equation}
\frac{\delta^*}{\delta} = \frac{h}{r}=\frac{1}{r}\sqrt{{r^2} - \frac{\delta^2}{4}}= \sqrt{1- \frac{\delta^2}{4 r^2}}.
\end{equation}
This proves the first inequality in (\ref{gap}).

Next, suppose $\delta >r$.  Figure \ref{Fig7} considers several possible scenarios for this case. If in the triangle  $\triangle pvp'$ the angle
$\angle pvp'$ is acute,  the line segment $p'v$ must necessarily intersect $C'$.  This implies $p''$ is the bisector of a chord in $C'$, hence inside of $C'$.  If the angle $\angle pvp'$ is at least $\pi/2$, then $p''$ will coincide with $v$. Hence, proving the second inequality in (\ref{gap}).
\end{proof}


\begin{figure}[htpb]
	\centering
	\begin{tikzpicture}[scale=0.4]


\begin{scope}[red]
         \draw (0.0,0.0) circle (7.0);
		 \draw (7.0,0.0) circle (7.0);
		 \draw (0.0,0.0) circle (4.48);
\end{scope}
		
		\draw (0.0,0.0) -- (7.0,0.0) -- (-2.0,-4.0) -- cycle;
      \draw (0,0) -- (7,0) node[pos=0.5, above] {$r$};
      \draw (-2,-4) -- (0,0) node[pos=0.5, above] {$\delta$};
       \draw (0,0) -- (1.15,-2.6) node[pos=0.5, right] {$\delta'$};
       \draw (1.15,-2.6) node[below] {$p''$};
       \filldraw (1.15,-2.6) circle (2pt);
		\draw (0,0) node[left] {$p$};
		\draw (7,0) node[right] {$v$};
		\draw (-2,-4) node[below] {$p'$};
         \draw (14,0) node[right]{$C'$};
          \draw (-7,0) node[left]{$C$};
           \draw (-4.48,0) node[left]{$C''$};
           \filldraw (0,0) circle (2pt);
\filldraw (7,0) circle (2pt);
\filldraw (-2,-4) circle (2pt);
		
	\end{tikzpicture}
\begin{center}
\caption{Depiction of gaps $\delta=d(p',p)$, $\delta'=d(p'',p)$, when $\delta \leq r=d(p,v)$.} \label{Fig5}
\end{center}
\end{figure}
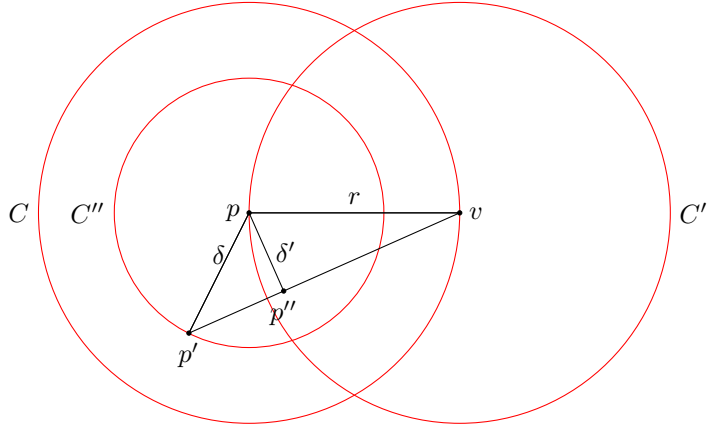

\section{The Triangle Algorithm and Its  Analysis}  \label{sec5}
In this section we describe a simple algorithm for solving problem (P). For convenience we shall refer to this algorithm as the {\it Triangle Algorithm}.  The justification in the name lies in the fact that in each iteration the algorithm searches for a triangle $\triangle pp'v_j$ where $v_j \in S$, $p' \in conv(S) \setminus \{p\}$, such that $d(p',v_j) \geq d(p,v_j)$. Given that such triangle exists, it uses $v_j$ as a $p$-pivot to ``pull'' the current iterate $p'$ closer to $p$ to get a new iterate $p'' \in conv(S)$.  If no such a triangle exists, then by Theorem \ref{thm2}, $p'$ is a $p$-witness certifying that $p$ is not in $conv(S)$. The steps of the algorithm are described in the box. Note that given the coordinates of $p'$ and $\alpha_i$'s that give its representation as a convex combination of the $v_i$'s, it takes $O(m)$ operations to compute $p''$ and $O(n)$ operations to compute the $\alpha'_i$'s.

\begin{center}
\begin{tikzpicture}
\node [mybox] (box){%
    \begin{minipage}{0.9\textwidth}
{\bf  Triangle Algorithm ($S=\{v_1, \dots, v_n\}$, $p$, $\epsilon \in (0,1)$)}\
\begin{itemize}

\item
{\bf Step 0.} ({\bf Initialization})  Let $p'=v={\rm argmin}\{ d(p,v_i): v_i \in S\}$.

\item
{\bf Step 1.} If $d(p,p') < \epsilon d(p,v)$, output $p'$ as $\epsilon$-approximate solution, stop. Otherwise, if there exists a pivot $v_j$, replace $v$ with $v_j$. If no pivot exists, then output $p'$ as a witness, stop.

\item
{\bf Step 2.}
Compute the {\it step-size}
\begin{equation} \label{pdp}
\alpha = \frac{(p-p')^T(v_j-p')}{d^2(v_j,p')}.
\end{equation}
Given the current iterate $p'=\sum_{i=1}^n \alpha_i v_i$, set the new {\it iterate} as:
\begin{equation}
p'' =
(1-\alpha)p' + \alpha v_j= \sum_{i=1}^n \alpha'_i v_i, \quad  \alpha'_j=(1-\alpha)\alpha_j+\alpha,  \quad \alpha_i= (1-\alpha)\alpha_i,  \quad \forall i \not =j.
\end{equation}
Replace $p'$ with $p''$,  $\alpha_i$ with $\alpha'_i$, for all $i=1,\dots,n$. Go to Step 1.
\end{itemize}

    \end{minipage}};
\end{tikzpicture}
\end{center}

By an easy calculation that shift $p'$ to the origin, it follows that
the point $p''$ in Step 2 is the closest point to $p$ on the line $p'v_j$.  Since $p''$ is a convex combination of $p'$ and $v_j$ it will remain in $conv(S)$.  The algorithm replaces $p'$ with $p''$ and repeats the above iterative step.  Note that a $p$-pivot  $v_j$ may or may not be a vertex of $conv(S)$. In the following we state some basic properties of the algorithm to be used in the analysis of its complexity.

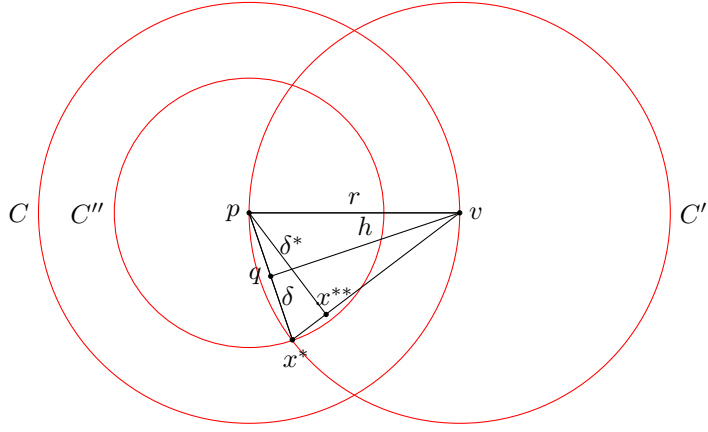
\begin{figure}[htpb]
	\centering
	\begin{tikzpicture}[scale=0.4]


\begin{scope}[red]
         \draw (0.0,0.0) circle (7.0);
		 \draw (7.0,0.0) circle (7.0);
		 \draw (0.0,0.0) circle (4.48);
\end{scope}
		
		\draw (0.0,0.0) -- (7.0,0.0) -- (1.44,-4.22) -- cycle;
      \draw (0,0) -- (7,0) node[pos=0.5, above] {$r$};
       \filldraw (0,0) circle (2pt);
       \filldraw (7,0) circle (2pt);
       \draw (.72,-2.11) -- (7,0) node[pos=0.5, above] {$h$};
       \draw (2.56,-3.38) -- (0,0) node[pos=0.5, above] {$~\delta^*$};
       \filldraw (0.72,-2.11) circle (2pt);
\filldraw (2.56,-3.38) circle (2pt);
       \draw (1.44,-4.22) -- (0,0) node[pos=0.2, above] {$~\delta$};
       \draw (2.56,-3.38) node[above] {$~~x^{**}$};
       \filldraw (1.44,-4.22) circle (2pt);
\filldraw (1.44,-4.22) circle (2pt);
       \draw (1.44,-4.22) node[below] {$~x^*$};
        \draw (.72,-2.11) node[left] {$q$};
		\draw (0,0) node[left] {$p$};
		\draw (7,0) node[right] {$v$};
         \draw (14,0) node[right]{$C'$};
          \draw (-7,0) node[left]{$C$};
           \draw (-4.48,0) node[left]{$C''$};
		
	\end{tikzpicture}
	
\begin{center}
\caption{The worst-case scenario for the gap $\delta'= \delta^*$, when $\delta \leq r$.} \label{Fig6}
\end{center}
\end{figure}

We are now ready to analyze the complexity of the algorithm. Set
\begin{equation} \label{R}
R=\max \bigg \{d(p,v_i), i=1, \dots, n  \bigg \}.
\end{equation}

\begin{cor}   \label{cor3}  Let $p,p',p'',v$ be as in Theorem \ref{thm3}, and  $r=d(p,v)$,   $\delta=d(p,p')$ , $\delta'=d(p,p'')$.  If $p'' \not = v$ and $\delta \leq r$, then
\begin{equation} \label{bound}
\delta'  \leq  \delta \sqrt{1- \frac{\delta^2}{4 R^2}} < \delta \exp \bigg (-\frac{\delta^2}{8 R^2} \bigg ).
\end{equation}
\end{cor}

\begin{proof} The first inequality follows from Theorem \ref{thm3} and the definition of $R$.  To prove the next inequality, we use that when $x \not =0$,
$1+x < \exp(x)$, and set $x= -{\delta^2}/{4 R^2}$.
\end{proof}

\begin{figure}[htpb]
	\centering
	\begin{tikzpicture}[scale=0.4]
\begin{scope}[red]
         \draw (0.0,0.0) circle (7.0);
		 \draw (7.0,0.0) circle (7.0);
\end{scope}
		\draw (0.0,0.0) -- (7.0,0.0) -- (-7, -10) -- cycle;
\draw (0.0,0.0) -- (7.0,0.0) -- (10, -11) -- cycle;
\draw (0.0,0.0) -- (7.0,0.0) -- (3, -10) -- cycle;
\draw (-7,-10) node[below] {$p'$};
\draw (10,-11) node[below] {$p'$};
\draw (3,-10) node[below] {$p'$};
\filldraw (-7,-10) circle (2pt);
\filldraw (10,-11) circle (2pt);
\filldraw (3,-10) circle (2pt);
\filldraw (0,0) circle (2pt);
\filldraw (7,0) circle (2pt);
\draw (-7,-10) -- (0,0) node[pos=0.5, left] {$\delta$};
\draw (10,-11) -- (0,0) node[pos=0.3, left] {$\delta$};
\draw (3,-10) -- (0,0) node[pos=0.4, left] {$\delta$};
\draw (10,-11) -- (7,0) node[pos=0.3, right] {~$\delta'$};
      \draw (0,0) -- (7,0) node[pos=0.5, above] {$r$};
       \draw (2.3,-3.38) -- (0,0) node[pos=0.3, left] {$\delta'$};
       \draw (2.3,-3.38) node[below] {$~~p''$};
       \filldraw (2.3,-3.38) circle (2pt);
        \filldraw (6,-2.5) circle (2pt);
        \draw (6,-2.5) node[below] {$~~p''$};
        \draw (6,-2.5) -- (0,0) node[pos=0.3, left] {$~\delta'$};
		\draw (0,0) node[left] {$p$};
		\draw (7,0) node[right] {$v$};
         \draw (14,0) node[right]{$C'$};
          \draw (-7,0) node[left]{$C$};

	\end{tikzpicture}
\begin{center}
\caption{Depiction of gaps $\delta=d(p',p)$, $\delta'=d(p'',p)$, when $\delta > r=d(p,v)$.} \label{Fig7}
\end{center}
\end{figure}


\begin{lemma}  \label{lem1}  Assume $p$ is in  $conv(S)$. Pick $p_0 \in conv(S)$, let $\delta_0=d(p_0,p)$. Assume $\delta_0 \leq R_0= \min \{d(p,v_i), i=1, \dots, n\}$.

Let $k\equiv k(\delta_0)$ be the maximum number of iterations of the Triangle Algorithm in order to compute $p_k \in conv(S)$ so that if $\delta_j=d(p_j,p)$ for $j=1, \dots, k$, we have
\begin{equation} \label{ineqaa}
\delta_k < \frac{\delta_0}{2} \leq \delta_j,  \quad j=1, \dots, k-1.
\end{equation}
Then, $k$ satisfies
\begin{equation} \label{iter}
k= k(\delta_0) \leq  \lceil N_0 \rceil, \quad   N_0 \equiv N(\delta_0) =(32 \ln 2)  \frac{R^2}{\delta_0^2} <  23 \frac{R^2}{\delta_0^2}
\end{equation}
\end{lemma}

\begin{proof}  For each $j=1, \dots, k-1$, Corollary \ref{cor3} applies.   The repeated application of  (\ref{bound}) in  Corollary \ref{cor3}, the assumption  (\ref{ineqaa}) that for each such $j$, $\delta_j \geq  \delta_0/2$, and the monotonicity of the exponential function implies
\begin{equation}
\delta_j < \delta_{j-1} \exp \bigg ({-\frac{\delta^2_{j-1}}{8 R^2}} \bigg ) \leq \delta_{j-1} \exp \bigg ({-\frac{\delta^2_{0}}{32 R^2}} \bigg ) < \delta_{j-2} \exp \bigg ({-\frac{2 \delta^2_{0}}{32 R^2}} \bigg ) \leq \cdots
\end{equation}
It follows that
\begin{equation} \label{x}
\delta_{k-1} < \delta_{0} \exp \bigg ({-\frac{(k-1) \delta^2_{0}}{32 R^2}} \bigg ).
\end{equation}
Thus from (\ref{bound}) and (\ref{x}) we have
\begin{equation}
\delta_k < \delta_{k-1} \exp \bigg ({-\frac{\delta_{k-1}^2}{8 R^2}} \bigg ) \leq \delta_0 \exp \bigg ({-\frac{k\delta_0^2}{32 R^2}} \bigg ).
\end{equation}
To have $\delta_k < \delta_0/2$, it suffices to satisfy
\begin{equation}
\exp \bigg ({-\frac{k\delta_0^2}{32 R^2}} \bigg ) \leq \frac{1}{2}.
\end{equation}
Solving for $k$ in the above inequality implies the claimed bound in (\ref{iter}).
\end{proof}

\begin{thm}   \label{thm4} The Triangle Algorithm correctly solves  problem (P) as follows:

(i) Suppose $p \in conv(S)$. Given $\epsilon >0$, $p_0 \in conv(S)$, with $\delta_0=d(p,p_0) \leq R_0 = \min \{d(p,v_i): i=1, \dots, n \}$. The number of iterations $K_\epsilon$ to compute a point $p_\epsilon$ in $conv(S)$ so that  $d(p,p_\epsilon) < \epsilon d(p, v_i)$, for some $v_i \in S$ satisfies
\begin{equation} \label{iter1}
K_\epsilon \leq   \frac{48}{\epsilon^2} =  O(\epsilon^{-2}).
\end{equation}
(ii) Suppose $p \not \in conv(S)$. If $\Delta$ denotes the distance from $p$ to $conv(S)$, i.e.
\begin{equation}
\Delta= \min \bigg \{d(x, p) : \quad x \in conv(S) \bigg \},
\end{equation}
the number of iterations $K_\Delta$ to compute a $p$-witness, a point
$p_\Delta$ in $conv(S)$ so that $d(p_\Delta,v_i) < d(p, v_i)$ for all $v_i \in S$, satisfies
\begin{equation} \label{iter2}
K_\Delta  \leq  \frac{48R^2}{\Delta^2} = O \bigg (\frac{R^2}{\Delta^2} \bigg).
\end{equation}
\end{thm}
\begin{proof} From Lemma \ref{lem1} and definition of $k(\delta_0)$ (see (\ref{iter})), in order to half the initial gap from $\delta_0$ to $\delta_0/2$, in the worst-case  the Triangle Algorithm requires $k(\delta_0)$ iterations.  Then, in order to reduce the gap from $\delta_0/2$ to $\delta_0/4$ it requires at most $k(\delta_0/2)$ iterations, and so on. From (\ref{iter}), for each nonnegative integer $r$  the worst-case number of iterations to reduce a gap from
$\delta_0/2^r$ to  $\delta_0/2^{r+1}$ is given by
\begin{equation}
k \bigg (\frac{\delta_0}{2^r} \bigg ) \leq   \bigg \lceil N \bigg (\frac{\delta_0}{2^r} \bigg ) \bigg  \rceil = \lceil 2^{2r} N_0 \rceil \leq 2^{2r}  \lceil N_0  \rceil.
\end{equation}
Therefore, if $t$ is the smallest index such that $\delta_0/2^{t} < \epsilon R$,  i.e.
\begin{equation}
2^{t-1} \leq \frac{\delta_0}{R \epsilon} <  2^t,
\end{equation}
then the total number of iterations of the algorithm, $K_\epsilon$,  to test if condition (i) is valid satisfies:
\begin{equation}
K_\epsilon \leq  \lceil N_0  \rceil (1+2^2+2^4 + \dots + 2^{2(t-1)}) \leq   \lceil N_0  \rceil  \frac{2^{2t} -1}{3} \leq   \lceil N_0  \rceil  2 \times 2^{2(t-1)} \leq (N_0+1) \frac{2\delta^2_0}{ R^2 \epsilon^2}
\end{equation}
From (\ref{iter}) we get

\begin{equation}
K_\epsilon \leq \bigg (23 \frac{ R^2}{\delta^2_0} +1 \bigg) \frac{2\delta^2_0}{ R^2 \epsilon^2} = \bigg (23 +  \frac{\delta^2_0}{ R^2} \bigg ) \frac{2}{ \epsilon^2}.
\end{equation}
Since $p \in conv(S)$ and from the definition of $R$ (see (\ref{R})) we have $\delta_0=d(p,p_0) \leq R$, hence we get the claimed bound on $K_\epsilon$  in (\ref{iter1}).

Suppose $p \not \in conv(S)$.  If suffices to choose $\epsilon=\Delta/R$. Then from (\ref{iter1}) and the definition of $\Delta$ we get the bound on
$K_\Delta$ in (\ref{iter2}).
\end{proof}

\section{Strict Distance Duality and A Minimax Duality} \label{sec6}

Here we first prove a stricter version of the distance duality, giving a stronger distance duality when $p \in conv(S)$. First we give a definition.

\begin{definition} Given $p' \in conv(S)$, we say $v_j \in S$ is a {\it strict pivot}
relative to $p$ (or {\it strict} $p$-pivot, or simply {\it strict} pivot) if $\angle p'pv_j \geq \pi/2$ (see Figure \ref{FigT}).
\end{definition}

\begin{figure}[htpb]
	\centering
	\begin{tikzpicture}[scale=0.6]	
      \draw (-4,0) -- (5,0) node[pos=1.0, right] {$\overline H$};
      \begin{scope}[red]
      \end{scope}[red]
       \filldraw (-4,0) circle (2pt);
       \filldraw (4,0) circle (2pt);
       \filldraw (-4,-3) circle (2pt);
		\draw (-4,0) node[left] {$p$};
\draw (-4,-3) node[left] {$p'$};
\filldraw (-3,-2.6) circle (2pt);
\draw (-4,-0) -- (-3,-2.6) node[pos=1.0, below] {$p''$};
		\draw (4,0) node[above] {$v_j$};
\draw (0,1) node[above] {$\overline H_+$};
\draw (3,1) node[right] {$v_j$};
\draw (-3,2) node[right] {$v_j$};
\filldraw (3,1) circle (2pt);
\filldraw (-3,2) circle (2pt);
\draw (-4,-3) -- (3,1) node[pos=0.5, right] {};
\draw (-4,-3) -- (-3,2) node[pos=0.5, right] {};
\draw (-4,-3) -- (4,0) node[pos=1.0, right] {};
\draw (-4,0) -- (3,1) node[pos=0.5, right] {};
\draw (-4,0) -- (-3,2) node[pos=0.5, right] {};
           \draw (-4,0) -- (-4,-3) node[pos=1.0, right] {};
           \draw (-4,-1.5) -- (5,-1.5) node[pos=1.0, right] {$H$};
	\end{tikzpicture}
\begin{center}
\caption{Several examples of strict $p$-pivot.} \label{FigT}
\end{center}
\end{figure}

\begin{thm}  \label{thm3dualitystronger} {\bf (Strict Distance Duality)} Assume $p \not \in S$.  Then $p \in conv(S)$ if and only if for each  $p' \in conv(S)$ there exists strict $p$-pivot, $v_j$. In particular,
\begin{equation} \label{inner}
(p'-p)^T(v_j-p) \leq 0.\end{equation}
\end{thm}

\begin{proof} Suppose $p \in conv(S)$.  Consider the orthogonal bisecting hyperplane to the line $p'p$, $H$, and the hyperplane parallel to it passing through $p$, $\overline H$, see Figure \ref{FigT}.  Let $\overline H_+$ be the halfspace determined by this hyperplane that excludes $p'$.  We claim it must contain a point $v_j \in S$ (see Figure \ref{FigT}). Otherwise, $p$ must be an extreme point of $conv(S \cup \{p\})$, but since $p \not \in S$, this implies $p \not \in conv(S)$.  This implies $\angle p'pv_j$ is at least $\pi/2$. Hence, (\ref{inner}) is satisfied.

Conversely, suppose that for each $p' \in conv(S)$, there exists $v_j \in conv(S)$ such that $\angle p'pv_j$ is at least $\pi/2$.  If $p \not \in conv(S)$, consider a witness $p'$. Then for each $v_j \in S$, $\angle p'pv_j < \pi/2$, a contradiction.
\end{proof}

\begin{thm} Given $p' \in conv(S)$, suppose $v_j$ is a strict $p$-pivot. Then if $\delta=d(p,p')$, $r=d(p,v_j)$, and $\delta'=d(p,p'')$, $p''$ the nearest to $p$ on the line segment  $p'v_j$, we have
\begin{equation}
\delta' = \frac{\delta r}{\sqrt{r^2+\delta^2}}  \leq
\delta \sqrt{1-\frac{\delta^2}{2r^2}}  \leq \delta \exp \bigg (-\frac{\delta^2}{4r^2} \bigg ).\end{equation}
\end{thm}
\begin{proof}  It is easy to see that for a given strict pivot $v_j$, the worst-case of error occurs when $\angle p'pv_j=\pi/2$, (see Figure \ref{FigT}).  Now using the similarity of the triangles $\triangle p'pv_j$ and $\triangle pp''p'$ (see Figure \ref{FigT}), we get the equality in the theorem.  To prove the first inequality, we use the fact that for a positive number  $x <1$,
\begin{equation}
\frac{1}{1+x} \leq 1- \frac{x}{2},\end{equation}
and let $x= \delta^2/r^2$. The second inequality follows from the inequality $1-x \leq \exp(x)$.
\end{proof}

Thus when $p \in conv(S)$ using a strict pivot we get a better constant in the worst-case complexity of the Triangle Algorithm than using any pivot.

\begin{definition} We say $p' \in conv(S)$ is a {\it strict witness}
relative to $p$ (or simply {\it strict} witness) if there is no strict pivot at $p'$. Equivalently, $p'$ is a strict witness if the orthogonal hyperplane to the line $p'p$ at $p$ separates $p$ from $conv(S)$. We denote the set of all strict witnesses by $\widehat W_p$.
\end{definition}

Clearly $\widehat W_p$ contains $W_p$ (see Definition \ref{defn4}). However, interestingly while $W_p$ is not described by a set of linear inequalities, $\widehat W_p$ can be characterized by a set of strict linear inequalities.  The following is straightforward.

\begin{prop} We have
\begin{equation}
\widehat W_p = \bigg \{x \in conv(S):  (x-p)^T (v_i-p) > 0, i=1, \dots, n \bigg \}.  \qed
\end{equation}
\end{prop}

Clearly $p \in conv(S)$ if and only if $\widehat W_p$ is empty. Figure \ref{Fig3Witness} shows the difference between $W_p$ and $\widehat W_p$. Its witness set was considered in an earlier example. The figure suggest that when $p \not \in conv(S)$, using a strict pivot would detect the infeasibility of $p$ sooner than using any pivot.

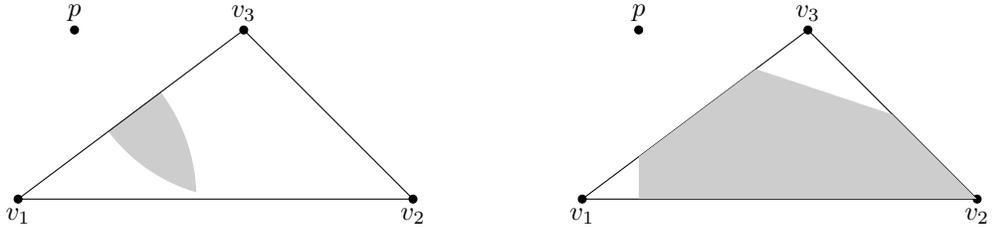
\begin{figure}[htpb]
	\centering
	
	\begin{tikzpicture}[scale=0.75]
		\begin{scope}[red]
		 \clip  (0.0,0.0) circle (3.162);
		 \clip (7.0,0.0) circle (6.7);
		 \clip (4.0,3.0) circle (3);
          \clip (0.0,0.0) -- (7.0,0.0) -- (4.0,3.0) -- cycle;
\fill[color=gray!39] (-10, 10) rectangle (10, -10);
\end{scope}

\begin{scope}[red]
\end{scope}

\draw (0.0,0.0) -- (7.0,0.0) -- (4.0,3.0) -- cycle;
		\draw (0,0) node[below] {$v_1$};
		\draw (7,0) node[below] {$v_2$};
		\draw (4,3) node[above] {$v_3$};
\filldraw (0,0) circle (2pt);
\filldraw (7,0) circle (2pt);
\filldraw (4,3) circle (2pt);
\filldraw (1,3) circle (2pt) node[above] {$p$};



\draw (10.0,0.0) -- (17.0,0.0) -- (14.0,3.0) -- cycle;
		\draw (10,0) node[below] {$v_1$};
		\draw (17,0) node[below] {$v_2$};
		\draw (14,3) node[above] {$v_3$};
\filldraw (10,0) circle (2pt);
\filldraw (17,0) circle (2pt);
\filldraw (14,3) circle (2pt);
\filldraw (11,3) circle (2pt) node[above] {$p$};

\begin{scope}[red]
          \clip (11.0, .75) -- (11.0,0.0) -- (17.0,0.0) -- (14.0,3.0) -- cycle;
           \clip (11.0, .75) -- (11.0,0.0) -- (17.0,0.0) -- (15.5,1.5) -- (11, 3) --cycle;
\fill[color=gray!39] (-10, 20) rectangle (20, -10);
\end{scope}

	\end{tikzpicture}
	\caption{Witness set $W_p$ (left) and Strict Witness set $\widehat W_p$ (right).}
\label{Fig3Witness}
\end{figure}

Let us consider the strict distance duality for the case when $p=0$, and let $A=[a_1, \dots, a_n]$ be the matrix of the points in $S$.  This can always be assumed since we can shift each of the points by $p$. In this case $p \in conv(S)$ if and only if $Ax=0, e^Tx=1, x \geq 0$ is feasible. The corresponding strict duality then implies given $x \in \Sigma_n= \{x: e^Tx=1, x \geq 0\}$,  there exists $a_j$ such that $a_j^T Ax \leq 0$.
This together with the fact that minimum of a linear function over the simplex $\Sigma_n$ is attained at a vertex, implies the following: given  $x \in \Sigma_n$, there exists $a_j$ such that
\begin{equation}
\min_{y \in \Sigma_n}  y^TA^T Ax =a_j^T Ax \leq 0.
\end{equation}
Combining this with von Neumann's minimax theorem we can thus state the following

\begin{thm}  The set $\{x: Ax=0, x \in \Sigma_n\}$ is feasible if and only if
$$w_*= \max_{x \in \Sigma_n}  \min_{y \in \Sigma_n}  y^TA^TAx  =\min_{y \in \Sigma_n} \max_{x \in \Sigma_n}  y^TA^TAx  \leq 0.$$
\end{thm}

The quality $w_*$ is the value of the matrix $Q=A^TA$.  If $w_* <0$ then it can be shown that there exists a number $\lambda_* \in (0,1)$ such that in each iteration of the Triangle Algorithm we have $\delta_{k+1} \leq \lambda_* \delta_{k}$ (see (\ref{lambda})).

\section{Alternate Complexity Bound for Triangle Algorithm} \label{sec7}

Throughout the section we assume $p$,  $S$ and $R$ are as defined previously, and $\epsilon \in (0,1)$.

\begin{definition} Given $\epsilon \in (0,1)$, let
\begin{equation}
\theta^*= \max \{\theta'=\angle vp'p: \quad  v \in S, \quad p' \in conv(S),  \quad d(p',v) \geq d(p,v), \quad d(p,p') \geq \epsilon R \}.\end{equation}

For a given  $p' \in conv(S)$  satisfying $d(p,p') \geq \epsilon R$,  set
\begin{equation}
\theta_{p'}= \min \{\theta'=\angle vp'p: \quad  v \in S,  \quad d(p',v) \geq d(p,v)\}.\end{equation}

Let
\begin{equation}
 \theta_*= \max \{ \theta_{p'}: \quad p' \in conv(S),  \quad d(p',v) \geq d(p,v), \quad d(p,p') \geq \epsilon R \}.\end{equation}

Let
\begin{equation}
\widehat \theta_*= \sup \{ \theta_{p'}: \quad p' \in conv(S),  \quad d(p',v) \geq d(p,v) \}.\end{equation}
\end{definition}

To describe $\theta^*$ and $\theta_*$  geometrically and the corresponding  visibility constants, let $B_{\epsilon R}(p)$ be the open ball of radius $\epsilon R$ at $p$.  Then $\theta^*$ is the maximum of all pivot angles as $p'$ ranges in $conv(S) \setminus B_{\epsilon R}(p)$. For each iterate $p' \in conv(S) \setminus B_{\epsilon R}(p)$, let $v_{p'}$ denote a pivot where the pivot angle  $\angle pp'v_{p'}$ is the least among all such pivots.  We refer to $v_{p'}$ as the {\it best pivot} at $p'$. Also,  $\theta_*$ is the maximum of these angles as $p'$ ranges in $conv(S) \setminus B_{\epsilon R}(p)$. Then
$\widehat \theta_*$ is the supremum of these angles as $p'$ ranges in $conv(S)$. The supremum is not necessarily attained, e.g. the case where $conv(S)$ is a square and $p$ lies on an edge.  In general because $\theta^*, \theta_*$ depend on $\epsilon$, they lie in $[0, \pi/2)$ but $\widehat \theta_*$ lies in $[0, \pi/2]$. For instance, in the example of square mentioned above with $p$ on an edge $\widehat \theta_*=\pi/2$. The farther away these angles are from $\pi/2$, the more effective steps the Triangle Algorithm will take in each iteration.

\begin{prop}  We have

(i)
\begin{equation}
\theta_* \leq \theta^*.\end{equation}

(ii)
\begin{equation}
\sin \theta_* \leq \sin \theta^* \leq \frac{1}{\sqrt{1+ \epsilon^2}}.\end{equation}
\end{prop}

\begin{proof}  (i) is immediate from the definitions of $\theta_*$ and $\theta^*$.  Since both angles are acute the first inequality in (ii) follows.  To prove the second inequality in (ii), consider the triangle $\triangle pp'v$ (see Figure \ref{Fig5zz}). Since $\angle p'pv$ is non-acute it is easy to argue
\begin{equation}
\sin \theta' \leq \frac{d(p,v)}{\sqrt{d^2(p,v) + d^2(p,p')}}.\end{equation}
Dividing the numerator and denominator of the right-hand-side of the above by $d(p,v)$, and using that $d(p,p')/d(p,v) \geq d(p,p')/R \geq  \epsilon$, we get the second inequality in (ii).
\end{proof}

\begin{definition} Given $p$, $S=\{v_1, \dots, v_n\}$,   $A=[v_1, \dots, v_n]$, and   $\epsilon \in (0,1)$,   the {\it arbitrary-pivot visibility constant}, $\nu^*$, and the corresponding {\it arbitrary-pivot visibility factor}, $c^*=c^*[p,A, \epsilon]$, are determined from the equation
\begin{equation} \label{nustar}
\nu^*=\sin \theta^*= \frac{1}{\sqrt{1+c^*}}.
\end{equation}
The {\it best-pivot visibility constant}, $\nu_*$, and the corresponding {\it best-pivot visibility factor}, $c_*=c_*[p,A, \epsilon]$ are determined from the equation
\begin{equation} \label{nustar1}
\nu_*=\sin \theta_*= \frac{1}{\sqrt{1+c_*}}.
\end{equation}

The {\it absolute visibility constant}, $\widehat \nu_*$, and the corresponding {\it absolute visibility factor}, $\widehat c_*=\widehat c_*[p,A]$ are determined from the equation
\begin{equation}
\widehat \nu_*=\sin \widehat \theta_*= \frac{1}{\sqrt{1+\widehat c_*}}.\end{equation}

Let $\theta_p$ be the maximum of pivot-angles over a sequence of iterates of the Triangle Algorithm as applied to the given input data, $S$, $p$, $\epsilon$, before it terminates.  The {\it (observed) visibility constant}, $\nu$, and the corresponding {\it (observed) visibility factor}, $c=c[p,A, \epsilon]$ are determined from he equation
\begin{equation}
\nu=\sin \theta_p= \frac{1}{\sqrt{1+ c}}.
\end{equation}
\end{definition}

Using the inner product formula $\cos \theta= u^Tv/ \Vert u \Vert \Vert v \Vert$, where $\theta$ is the angle between  $ u, v \in \mathbb{R} ^m$, and setting $F(x)=
{(x-p)^T (x-v)}/{\Vert x - p \Vert  \Vert x - v \Vert }$, we may give alternative definition the above visibility parameters. For instance,  $\nu_*$ and $\widehat \nu_*$ can be defined as

\begin{equation}
\nu_* = \max_{\{x \in conv(S),~ d(x,p) \geq \epsilon R\}}~ \min_{\{v \in S,~ d(x,v) \geq d(p,v) \}} \sqrt{1- F(x)^2},
\end{equation}

\begin{equation}
\widehat \nu_* = \sup_{\{x \in conv(S)\}} ~\min_{\{v \in S,~ d(x,v) \geq d(p,v) \}} \sqrt{ 1- F(x)^2}. \end{equation}
Note that generally $\nu_*$ is  dependent on $\epsilon$ while $\widehat \nu_*$ is independent of $\epsilon$.

\begin{example}  Consider the case where $S$ is the vertices of a square and $p$ is the center, Figure \ref{FigTP}.  In this case $\theta_*=\pi/8$ and the absolute visibility constant satisfies $\widehat \nu_*=\sin \pi/8 \approx .3826$. By moving $p$ inside the square, $c(p,A, \epsilon)$ varies. However, not drastically if it stays reasonably away from the boundary. The visibility constants become worst when $p$ is at the midpoint of one of the edges. The absolute visibility constant in this case is at its worst, equal to one.

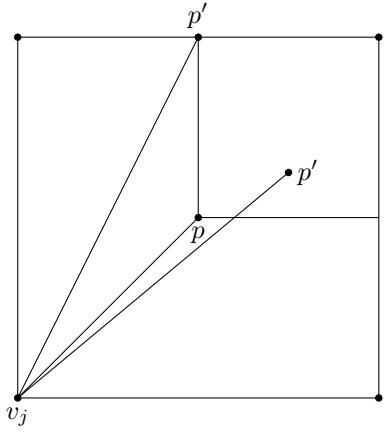
\begin{figure}[htpb]
	\centering
	\begin{tikzpicture}[scale=0.6]	
      \begin{scope}[red]
      \end{scope}[red]
       \filldraw (-4,0) circle (2pt);
       \filldraw (4,0) circle (2pt);
         \filldraw (0,-4) circle (2pt);
         \filldraw (-4,-8) circle (2pt);
         \filldraw (0,0) circle (2pt);
         \filldraw (4, -8) circle (2pt);
		\draw (0,-4) node[below] {$p$};
\draw (0,0) node[above] {$p'$};
\draw (-4,-8) node[below] {$v_j$};
\filldraw (2,-3) circle (2pt);
\draw (0,0) -- (0,-4) node[pos=0.5, right] {};
\draw (0,-4) -- (4,-4) node[pos=0.5, right] {};
\draw (-4,0) -- (4,0) node[pos=0.5, right] {};
\draw (-4,0) -- (-4,-8) node[pos=0.5, right] {};
\draw (-4,-8) -- (4,-8) node[pos=1.0, right] {};
\draw (-4,-8) -- (2,-3) node[pos=1.0, right] {$p'$};
\draw (4,-8) -- (4,0) node[pos=0.5, right] {};
\draw (0,0) -- (-4,-8) node[pos=0.5, right] {};
\draw (0,-4) -- (-4,-8) node[pos=0.5, right] {};
	\end{tikzpicture}
\begin{center}
\caption{When $conv(S)$ is a square, $p$ the center,  $\theta_*=\pi/8$, $\widehat \nu_*\approx .38$.} \label{FigTP}
\end{center}
\end{figure}
\end{example}

We now prove an alternative complexity bound for the Triangle Algorithm.

\begin{thm}  Suppose $p \in conv(S)$. Let $\delta_0=d(p,p_0)$, $p_0 \in conv(S)$. The number of arithmetic  operations of the Triangle Algorithm to get $p_\epsilon \in conv(S)$ such that $d(p_\epsilon, p) < \epsilon R$, $R= \max\{d(p,v_j): v_j \in S\}$, having (observed) visibility constant $\nu$ and (observed) visibility factor $c$,  is
\begin{equation}
O \bigg (mn \frac{1}{c}\ln \frac{\delta_0}{\epsilon R} \bigg ). \end{equation}
In particular,  when the algorithm uses the arbitrary-pivot strategy  its complexity is
\begin{equation}
O \bigg (mn \frac{1}{c^*}\ln \frac{\delta_0}{\epsilon R} \bigg ),\end{equation}
and when it uses the best pivot-strategy its complexity is
\begin{equation}
O \bigg (mn \frac{1}{c_*} \ln \frac{\delta_0}{\epsilon R} \bigg ).\end{equation}
\end{thm}
\begin{proof}  With $\delta=d(p',p)$ and $\delta'=d(p'',p)$, where $p'$ and $p''$ are two consecutive iterates and $v$ the corresponding pivot, and $\theta'=\angle pp'v$, we have (see Figure \ref{Fig5zz})
\begin{equation}
 \frac{\delta'}{\delta}=\sin \theta' \leq \nu.\end{equation}

\begin{figure}[htpb]
	\centering
	\begin{tikzpicture}[scale=0.4]


\begin{scope}[red]
\end{scope}
		
		\draw (0.0,0.0) -- (7.0,0.0) -- (-2.0,-4.0) -- cycle;
      \draw (0,0) -- (7,0) node[pos=0.5, above] {};
      \draw (-2,-4) -- (0,0) node[pos=0.5, above] {$\delta$};
       \draw (0,0) -- (1.15,-2.6) node[pos=0.5, right] {$\delta'$};
       \draw (1.15,-2.6) node[below] {$p''$};
       \filldraw (1.15,-2.6) circle (2pt);
		\draw (0,0) node[left] {$p$};
		\draw (7,0) node[right] {$v$};
		\draw (-2,-4) node[below] {$p'$};
\draw (-1.5,-3.5) node[above] {$~~~~\theta'$};
           \filldraw (0,0) circle (2pt);
\filldraw (7,0) circle (2pt);
\filldraw (-2,-4) circle (2pt);
		
	\end{tikzpicture}
\begin{center}
\caption{Relationship between $\delta=d(p',p)$, $\delta'=d(p'',p)$, and $\theta'$.} \label{Fig5zz}
\end{center}
\end{figure}
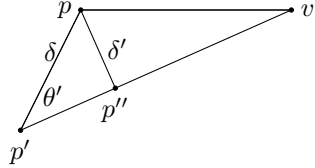

Thus if the iterates in the Triangle Algorithm are $p_i$, and $\delta_i=d(p_i,p)$, then after $k$ iterations we will have
\begin{equation}
\delta_k \leq  \nu^k \delta_0.\end{equation}
To have the right-hand-side to be bounded above by $\epsilon R$, $k$ must satisfy
\begin{equation}
k \ln \nu = k \ln \frac{1}{\sqrt{1+c}} \leq \ln \frac{\epsilon R}{\delta_0}.\end{equation}
Equivalently,
\begin{equation}
\frac{k}{2} \ln (1+c)  \geq \ln \frac{\delta_0}{\epsilon R}.\end{equation}
For $u \in (0,1)$  we may write
\begin{equation}
\ln(1+u) \geq \frac{u}{2}.\end{equation}
Thus it suffices to choose $k$ satisfying
\begin{equation}
\frac{k}{4} c  \geq \ln \frac{\delta_0}{\epsilon R}.\end{equation}
Each iteration takes at most $O(mn)$ arithmetic operations. Under arbitrary pivot strategy we have, $\nu \leq \nu^*$ and under the best pivot strategy we have $\nu \leq \nu_*$. These imply, $c^* \leq c$ and $c_* \leq c$, respectively. Hence the claimed complexity bounds.
\end{proof}

\begin{remark}  Let us assume $\epsilon =2^{-L}$ for some natural number $L$. We examine the number of iterations of the Triangle Algorithm to compute $p_k$ so that  $\delta_k \leq 2^{-L} \delta_0$.  Since $\delta_k \leq \nu^k \delta_0$, if the visibility constant satisfies $\nu=0.5$, then clearly the number of iterations is $L$.
Table \ref{table2} shows the number of such iterations when $\nu=.9$, $.99$, $.999$, $.9999$, and $.99999$. As we see the numbers are quite reasonable and independent of $m, n$.

\begin{table}[htpb]
	\renewcommand{\arraystretch}{1.0}
	\centering
\scalebox{0.95}{
\begin{tabular}{|l|c|}

\hline
Visibility constant $\nu$ &  Bound on no. of iterations $k$
\\
\hline

.9  &  7 $\times$ L
\\
.99 & 70 $\times$ L
\\
.999 & 700 $\times$ L
\\
.9999 & 7000 $\times$ L
\\
.99999 & 70000 $\times$ L
\\
\hline

	\end{tabular}
}
\caption{ Number of iterations of Triangle Algorithm for different visibility constants in order to obtain $\delta_k \leq 2^{-L} \delta_0$, independent of $m$ and $n$.}
	\label{table2}
\end{table}

\end{remark}

\begin{thm}  Assume $p$ lies in the relative interior of $conv(S)$. Let $\rho$ be the supremum of radii of the balls centered at $p$ in this relative interior.  Let $\delta_0=d(p_0,p)$, $p_0 \in conv(S)$.
Let $R=max\{d(p,v_i), i=1, \dots,n\}$. Given $\epsilon \in (0,1)$, suppose the Triangle Algorithm uses a strict pivot in each iteration. The number of arithmetic operations to compute $p' \in conv(S)$ such that $d(p',p) < \epsilon R$ satisfies
\begin{equation}
O\bigg (mn  \frac {R^2}{\rho^2}\ln \frac{\delta_0 }{\epsilon R} \bigg ), \quad \delta_0=d(p_0,p).\end{equation}
In particular, if $\rho \geq \epsilon^{1/t} R$, $t$ a natural number then the complexity is
\begin{equation}
O \bigg (mn \frac{1}{\sqrt[t]{\epsilon^{2}}} \ln \frac{\delta_0 }{\epsilon R} \bigg ).\end{equation}
\end{thm}

\begin{proof} Given the current iterate $p'$ where $d(p,p') \geq \epsilon R$,  let $q$ be the point on the extension of the line $p'p$, where $d(p,q)=\rho$, see Figure \ref{Figzz7}. Then we have $q \in conv(S)$. The orthogonal hyperplane to the line segment $pq$ must contain a  point $v$ in $S$ on the side that excludes $p$, see Figure \ref{Figzz7}. For such strict pivot $v$, $\angle pqv$ is non-acute.

\begin{figure}[htpb]
	\centering
	\begin{tikzpicture}[scale=0.6]	
      \draw (-4,-3) -- (8,1.5) node[pos=0.5, above] {$r$};
      \draw (-4,0) -- (8,1.5) node[pos=0.55, above] {};
      \draw (-4,-5) -- (8,1.5) node[pos=0.55, above] {};
      \draw (-4,1.5) -- (8,1.5) node[pos=0.55, above] {};
      \draw (-4,-3) circle (3);
      \filldraw (-4,1.5) circle (2pt);
      \filldraw (-4,-5) circle (2pt);
      \filldraw (-4,0) circle (2pt);
      \filldraw (-3.2,-4.55) circle (2pt);
      \draw (-3.2,-4.55) node[below] {$p''$};
      \draw (-4,1.5) node[above] {$s$};
      \draw (-4,-3) -- (-4,1.5);
       \draw (-4,-3) -- (-3.2,-4.55) node[pos=0.7, above] {$\delta'$};
      \begin{scope}[red]
      \end{scope}[red]
      \filldraw (8,1.5) circle (2pt);
       \filldraw (-4,-3) circle (2pt);
		\draw (-4,0) node[left] {$q$};
\draw (-4,-3) node[left] {$p$};
\draw (-4,-5) node[below] {$p'$};
 \draw (-4,-5) -- (-4,-3) node[pos=0.5, left] {$\delta$};
		\draw (8,1.5) node[right] {$v$};
           \draw (-4,0) -- (-4,-3) node[pos=0.5, left] {$\rho$};
           \draw (-4,-3) -- (4,0) node[pos=0.5, right] {};
	\end{tikzpicture}
\begin{center}
\caption{Existence of a strict pivot with visibility constant bounded by $1/\sqrt{1+\rho^2/R^2}$.} \label{Figzz7}
\end{center}
\end{figure}
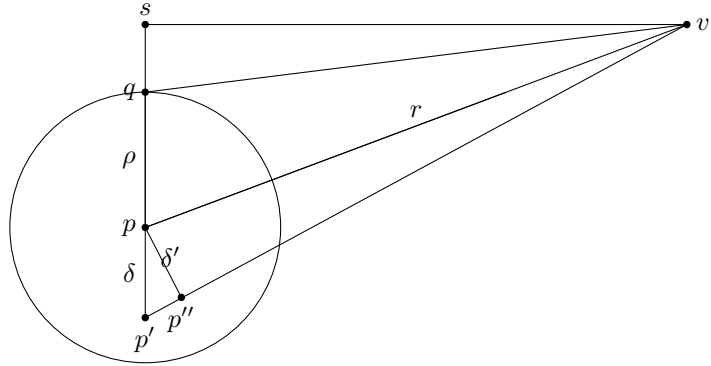

Let $p''$ be the next iterate.  Then we have (see Figure \ref{Figzz7}):
\begin{equation}
\sin \angle pp'v = \frac{\delta'}{\delta}.\end{equation}
Let $r=d(p,v)$. We claim
\begin{equation}
\sin \angle pp'v \leq  1/\sqrt{1+ \frac{\rho^2}{r^2}}.\end{equation}
To prove this, note that we have
\begin{equation}
\sin \angle pp'v \leq \sin \angle qpv.\end{equation}
However, considering that $\angle pqv$ is non-acute, we can introduce $s$ so that $\angle psv=\pi/2$. Thus we may write

\begin{equation}
\sin \angle qpv = \sin \angle spv = \frac{\sqrt{r^2- (d(q,s) +\rho)^2}}{r}   \leq \frac{\sqrt{r^2- \rho^2}}{r} = \sqrt{1- \frac{\rho^2}{r^2}} \leq 1/\sqrt{1 + \frac{\rho^2}{r^2}} \leq 1/\sqrt{1 + \frac{\rho^2}{R^2}}.\end{equation}
Thus if $\delta_k$ represents the gap in $k$ iterations and $\delta_0$ the initial gap, then we have
\begin{equation}
\delta_k  \leq \bigg ( 1/\sqrt{1 + \frac{\rho^2}{R^2}} \bigg )^k \delta_0.\end{equation}
To have $\delta_k < \epsilon R$ it suffices to find $k$ so that
\begin{equation}- \frac{k}{2} \ln \bigg (1+ \frac{\rho^2}{R^2} \bigg ) < \ln \frac{\epsilon R}{\delta_0}.\end{equation}
Equivalently,
\begin{equation} \frac{k}{2} \ln \bigg (1+ \frac{\rho^2}{R^2} \bigg ) > \ln \frac{\delta_0}{\epsilon R}.\end{equation}
As shown in previous theorem, for $u \in (0,1)$  we have, $\ln(1+u) \geq \frac{u}{2}$. Thus we choose $k$ satisfying
\begin{equation}k >  \frac{4R^2}{\rho^2} \ln \frac{\delta_0}{\epsilon R}.\end{equation}
\end{proof}

\begin{remark} According to the above theorem the best-pivot visibility constant, $\nu_*$,
will be good when $p$ lies  in a reasonable size ball in the relative interior of $conv(S)$.
For instance, in the example of a square in Figure \ref{FigTP}, aside from the case when $p$ is center of the square, for points in a reasonably large circle we would expect the complexity of the Triangle Algorithm to be very good.  Only points near a boundary line can slow down the algorithm. However, even for these points when iteration begins to slow down we can introduce {\it auxiliary pivots} to improve the visibility constant.  This will be considered in a subsequent section.
\end{remark}

We may summarize the performance of the Triangle Algorithm with the following complexity theorem.

\begin{thm}  \label{compthm} Let $\delta_0=d(p,p_0)$, $p_0 \in conv(S)$.  Suppose $p \in conv(S)$. The number of arithmetic  operations of the Triangle Algorithm to get $p_\epsilon \in conv(S)$ such that $d(p_\epsilon, p) < \epsilon R$, $R= \max\{d(p,v_j): v_j \in S\}$, where the  visibility factor is $c_1=c(p,A, \epsilon)$, satisfies
\begin{equation}
O\bigg (mn \min \bigg \{ \frac{1}{\epsilon^2}, \frac{1}{c_1}\ln \frac{\delta_0}{\epsilon R} \bigg \} \bigg ), \quad c_1\geq \epsilon^2.
\end{equation}
Suppose $p \not \in conv(S)$. The number of arithmetic  operations of the Triangle Algorithm to get a witness  $p' \in conv(S)$, where the  visibility factor is $c_2=c(p,A, \Delta/R)$  satisfies
\begin{equation}
O\bigg (mn \min \bigg \{ \frac{R^2}{\Delta^2}, \frac{1}{c_2}\ln \frac{\delta_0}{\Delta} \bigg \} \bigg ), \quad
c_2 \geq \frac{\Delta^2}{R^2}. \qed
\end{equation}
\end{thm}

We close this section by using the strict distance duality to associate
{\it best-strict pivot visibility constant} to a given input data $p$ and $S =\{a_1, \dots, a_n\}$. Consider the Triangle Algorithm where given an iterate $p' \in conv(S)$, it selects a pivot $v$ so that the angle $\angle p'pv$ is the largest.   Without loss of generality we assume  $p=0$ and let $A=[a_1, \dots, a_n]$.
By the strict duality theorem, given $x \in \Sigma_n=\{x: e^Tx=1, x \geq 0 \}$,  there exists $a_j$ such that $a_j^T Ax \leq 0$.

\begin{definition}  Given $p=0$ and $S=\{a_1, \dots, a_n\}$, $A=[v_1, \dots, a_n]$, the {\it best-strict pivot visibility constant} with respect to $p$ and $S$ is the number
\begin{equation} \label{lambda}
\lambda_*= \sqrt{1 - \phi_*^2},
\end{equation}
where
\begin{equation}
\phi_*= \max_{x \in \Sigma_n} \min_{a_j \in S}   \frac{a_j^T Ax}{\Vert a_j \Vert \Vert Ax \Vert}.
\end{equation}
\end{definition}

From the definition of  $\lambda_*$  and $\nu_*$ (see (\ref{nustar})), it is easy to conclude that $\nu_* \leq \lambda_*$.

\begin{prop} Suppose $p=0 \in conv(S)$. If $\lambda_* <1$, then in $k$  iteration of the Triangle Algorithm using strict pivot, the gap $\delta_k=d(p,p_k)$ will satisfy
\begin{equation}
\delta_{k} \leq \lambda_*^k \delta_0.
\end{equation}
\end{prop}
\begin{proof} It suffices to show that if $\delta=\Vert p'\Vert $ and $\delta'=\Vert p'' \Vert$ are two consecutive gaps in the Triangle Algorithm, computed based on a strict pivot $v$ at the current iterate $p'$, then
$\delta' \leq \lambda_* \delta$. By the definition of $\phi_*$ and the inner product formula for cosine, we must have $\cos (\angle p'pv) \leq \phi_*$.  Thus  $\sin (\pi - \angle p'pv) = \delta'_*/\delta$, for some $\delta'_*$ satisfying $\delta' \leq \delta'_*$, see Figure \ref{Fig9zz}. But from definition $\lambda_*= \sin (\pi - \angle p'pv)$.
\end{proof}

\begin{figure}[htpb]
	\centering
	\begin{tikzpicture}[scale=0.4]


\begin{scope}[red]
\end{scope}
		
		\draw (0.0,0.0) -- (7.0,0.0) -- (-2.0,-4.0) -- cycle;
      \draw (0,0) -- (7,0) node[pos=0.5, above] {};
      \draw (-2,-4) -- (0,0) node[pos=0.5, above] {$\delta$};
       \draw (0,0) -- (1.15,-2.6) node[pos=0.5, right] {$\delta'$};
       \draw (1.15,-2.6) node[below] {$p''$};
       \filldraw (1.15,-2.6) circle (2pt);
		\draw (0,0) node[left] {$p$};
\draw (0,0) node[below] {};
		\draw (7,0) node[right] {$v$};
		\draw (-2,-4) node[below] {$p'$};
   \draw (-2,-4) -- (5,-4) node[pos=0.5, right] {};
   \draw (0,0) -- (0,-4) node[pos=0.5] {$\delta'_*$};
\draw (-2,-4) node[right] {};
\draw (0,-4) node[below] {$p''_*$};
           \filldraw (0,0) circle (2pt);
\filldraw (7,0) circle (2pt);
\filldraw (-2,-4) circle (2pt);
		
	\end{tikzpicture}
\begin{center}
\caption{$\delta' < \delta'_*$, since $\angle p'pv$ is non-acute.} \label{Fig9zz}
\end{center}
\end{figure}
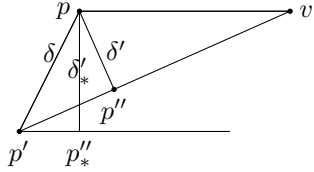

\section{Virtual Triangle Algorithm and Its Approximated Version} \label{sec8}

Here we will consider a version of the Triangle Algorithm that works with {\it virtual} iterates.  We call the iterates as  such because in each iteration we only know the coordinates of the iterates but not their representation as points in $conv(S)$. In fact they may not even lie in $conv(S)$.  We thus call the algorithm {\it Virtual Triangle Algorithm}.  Describing such an algorithm has two advantages: (i) it may quickly identify the case when $p \not \in conv(S)$; (ii) it gives rise to a version of the Triangle Algorithm, {\it Approximated Virtual Triangle Algorithm}, that has very efficient complexity, does provide an approximate solution in $conv(S)$ when $p \in conv(S)$, however it may not necessarily give an approximation to within a desired accuracy. It may also quickly provide a witness when $p \not \in conv(S)$. Consider the box below and
Figure \ref{Figzz9}.

\begin{center}
\begin{tikzpicture}
\node [mybox] (box){%
    \begin{minipage}{0.9\textwidth}
{\bf  Virtual Triangle Algorithm    ($S=\{v_1, \dots, v_n\}$, $p$, $\epsilon \in (0,1))$}\

\begin{itemize}
\item
{\bf Step 0.} ({\bf Initialization})  Let $p'=v={\rm argmin}\{ d(p,v_i): v_i \in S\}$.

\item
{\bf Step 1.} If $d(p,p') < \epsilon d(p,v)$, output $p'$ as $\epsilon$-approximate solution, stop. Otherwise, if there exists a pivot $v_j$, replace  $v$ with $v_j$.
If no pivot exists, output $p'$ as a {\it general $p$-witness}, stop.

\item
{\bf Step 2.}  Let
\begin{equation}
\overline v_j= p+\frac{d(p,p')}{d(p, v_j)} (v_j-p).\end{equation}
Replace $p'$ with $\overline p''$ below, and go to Step 1.
\begin{equation}
\overline p''= \frac{1}{2} p'+\frac{1}{2} \overline v_j.\end{equation}
\end{itemize}

    \end{minipage}};
\end{tikzpicture}
\end{center}

\begin{figure}[htpb]
	\centering
	\begin{tikzpicture}[scale=0.8]	
      \draw (-4,-3) -- (8,1.5) node[pos=0.55, above] {};
      \draw (-4,0) -- (8,1.5) node[pos=0.55, above] {};
      \draw (-4,0.) circle (3);
      \begin{scope}[red]
      \end{scope}[red]
      \filldraw (8,1.5) circle (2pt);
       \filldraw (-4,0) circle (2pt);
       \filldraw (-4,-3) circle (2pt);
       \filldraw (-1, .375) circle (2pt);
       \filldraw (-2.5, -1.3125) circle (2pt);
       \draw (-4,0) -- (-2.5, -1.3125) node[pos=0.3, right] {$\overline \delta'$};
        \draw (-2.5, -1.3125) node[right] {$\overline p''$};
       \draw (-1, .375) node[above right] {$\overline v_j$};
       \draw (-4,-3) -- (-1, .375);
       \filldraw (.25*12, .25*1.5) +(-4,0) circle (2pt);
       \filldraw (-2.15,-2.31) circle (2pt);
       \draw (-2.15,-2.31) -- (-2.5, -1.3125);
       \draw (-2.15,-2.31) node[below right] {$\overline p''(1)$};
		\draw (-4,0) node[left] {$p$};
\filldraw (-3.,-2.65) circle (2pt);
\draw (-3.0,-2.7) node[below] {$p''$};
\draw (-4,-3) node[below] {$p'$};
		\draw (8,1.5) node[right] {$v_j$};
           \draw (-4,0) -- (-4,-3) node[pos=0.5, left] {$\delta$};
             \draw (-4,0) -- (-3.,-2.65) node[pos=0.7, left] {$\delta'$};
           \draw (-4,-3) -- (4,0) node[pos=0.5, right] {};
	\end{tikzpicture}
\begin{center}
\caption{One iteration of Virtual Triangle Algorithm, $\overline p''$, and Approximated version $\overline p''(1)$.} \label{Figzz9}
\end{center}
\end{figure}

\begin{thm} If $p \in conv(S)$, the Virtual Triangle Algorithm in $O(mn \ln \epsilon^{-1})$ arithmetic operations generates a point $p' \in conv(S)$ such that
$d(p,p') < \epsilon R$. However, $p'$ has no certificate (i.e. the representation of $p'$ as a convex combination $v_i$'s is not  available).
\end{thm}

\begin{proof}  If $p \in conv(S)$, so is $\overline v_j$, wee Figure \ref{Figzz9}. In the worst-case the triangle $\triangle pp'\overline v_j$ is an equilateral.  From this it follows that it in each iteration we have
\begin{equation}
\delta' \leq \frac{\sqrt{3}}{2} \delta.\end{equation}
Hence, analogous to the analysis of previous theorem the proof of complexity is immediate.
\end{proof}

\begin{remark}  If $p \not \in conv(S)$ the Virtual Triangle Algorithm may still correctly determine this by finding a witness with respect to $\overline v_j$. For example, consider Figure \ref{Figzz9} and assume $S$ consists of $p'$, $v_j$ and  $\overline v_j$.  Then $p \not \in conv(S)$ and neither $p'$ nor $p''$ is a witness. However, $\overline p''$ is a witness.
\end{remark}

The Virtual Triangle Algorithm does not provide an approximation when $p \in conv(S)$, since it does not have a representation for  $\overline v_j$ as a point in $S$. Thus even if the iterate $\overline p''$ lies in $conv(S)$ it has no representation as a point in $conv(S)$. To remedy this we offer a variation where, starting with a $p''$, we estimate $\overline p''$, using the Triangle Algorithm itself. Consider the following algorithm.
\begin{center}
\begin{tikzpicture}
\node [mybox] (box){%
    \begin{minipage}{0.9\textwidth}
{\bf  Approximated  Virtual Triangle Algorithm} ($S=\{v_1, \dots, v_n\}$, $p$, $\epsilon \in (0,1)$, $t \in \mathbb{Z}_+$)\

\begin{itemize}
\item
{\bf Step 0.} ({\bf Initialization}) Let $p'=v={\rm argmin}\{ d(p,v_i): v_i \in S\}$.

\item
{\bf Step 1.} If $d(p,p') < \epsilon d(p, v)$, output $p'$ as $\epsilon$-approximate solution, stop. Otherwise, if there exists a pivot $v_j$, replace $v$ with $v_j$. If no such pivot exists, then output $p'$ as a  witness, stop.

\item
{\bf Step 2.}  Let $p''=\sum_{i=1}^n \alpha' v_i$ be as in Step 2 of Triangle Algorithm (see (\ref{pdp})).  Then, using $\overline p''(1)$, the nearest point to $p''$ on the line $p'v_j$, perform at most $t$ iterations of the Triangle Algorithm to test if $\overline p''$ lies in $conv(S)$.  If any of the  iterates, $\overline p''(r)$, $r \leq t$, is a witness that $\overline p''(r) \not \in conv(S)$, then $p \not \in conv(S)$, stop.

\item
{\bf Step 3.}
Replace $p'$ with $\overline p''(r)$ and go to Step 1.
\end{itemize}
\end{minipage}};
\end{tikzpicture}
\end{center}


The advantage of iterating towards $\overline p''$ as opposed to iterating directly toward
$p$ lies in that $\overline p''$, or an approximation to it, offers a new vantage point in order to get a good reduction toward $p$ itself in subsequent iterations. The following result assures that in each iteration the angle $pp'\overline p''(r)$ improves.

\begin{prop} For any $r \geq 2$,  $\angle pp'\overline p''(r) < \angle pp'p''$. Thus if $\overline p''(r)$ is used as a $p$-pivot, then $d(p,\overline p''(r)) < d(p,p'')$.
\end{prop}

\begin{proof}  We first note that $\overline p''(1)$ is nearest point to $p''$ on the line $p'v_j$.  However, for $r >1$, $p''(r)$  may not lie in the same Euclidean plane as the one shown in the Figure \ref{Figzz9}. We have
\begin{equation}
\angle pp'v_j = \angle pp' \overline v_j + \angle \overline v_jp'v_j.\end{equation}
On the other hand note that we must have
\begin{equation}
\angle pp'p''(r)  \leq  \angle pp' \overline v_j + \angle \overline v_jp'\overline p''(r).\end{equation}
However, since $d(\overline p''(r),\overline p'') < d(p'',\overline p'')$, it follows that
\begin{equation}
\angle \overline v_jp'\overline p''(r) <  \angle \overline v_jp'v_j.\end{equation}
\end{proof}

The following is straightforward and shows that the Approximated Virtual Triangle Algorithm may perform well if the visibility constant $\nu_*$ is a reasonable constant.

\begin{thm} Suppose $p \in conv(S)$, and $\epsilon = 2^{-L}$, where $L$ is a natural number. Suppose the Approximated Virtual Triangle Algorithm uses the best-pivot strategy. Then, in $O(mnLt)$ arithmetic operations it computes a point $p' \in conv(S)$ so that
\begin{equation}
d(p, p') \leq \nu_*^{tL} \delta_0, \quad \delta_0=d(p,p_0).\end{equation} \qed
\end{thm}

\section{Triangle Algorithm with Auxiliary Pivots} \label{sec9}

According to Theorem \ref{thm3}  the worst-case error bound of Triangle Algorithm for a given iterate $p' \in conv(S)$, and pivot $v_j \in S= \{v_1, \dots, v_n\}$ depends on $r=d(p,v_j)$ and $d(p,p')$. Specifically,  if $p''$ is the next iterate then $\delta'=d(p,p'')$ is bounded above by $\delta \sqrt{1- \delta^2/4r^2}$.  According to the alternate analysis, the ratio $\delta'/\delta$ is also related to the visibility constants.  These affect the complexity of the Triangle Algorithm. In this section we discuss how to introduce new points into $S$ so as to allow improving the visibility constant. In the Approximated Virtual Triangle Algorithm we already have made implicit use of such pivot, namely by introducing the point $\overline v_j$ (see Figure \ref{Figzz9}).

Suppose we have a finite subset $S'$ in $conv(S)$, distinct from $S$.  We refer to $S'$ as {\it auxiliary pivots}.  Suppose that in an iteration of the Triangle Algorithm the search for a pivot is carried out over $S \cup S'$.  Then it may be possible to improve $\delta'$ by using as pivot a point in $S'$. Not all such points may end up being used as  pivots.

Overall this increases $n=|S|$ to $n'=|S|+|S'|$, but it could drastically improve visibility constants. Here we discuss three strategies. Many other strategies are possible.

\begin{itemize}

\item
{\bf Strategy I.}   Given $p,p' \in conv(S)$ and a $p$-pivot $v$, if $v$ is a strict pivot, or if $d(p,p') \leq d(p,v)$ (see Figure \ref{Fig5zz}),  we can switch the role of $p'$ and $v$. In other words, $p'$ as an auxiliary point is a $v$-pivot, implying that the next iteration of the Triangle Algorithm, starting from $v$, is at least as effective.  Thus if $\delta'=d(p,p'')$ does not sufficiently reduce $\delta=d(p,p')$, we add $p'$ as auxiliary pivot and choose $v$ to be the next iterate. We may wish to delete $p'$, or keep it  as an auxiliary point for subsequent iterations. As an example, consider the case where $p$ is a midpoint of an edge on a square. This strategy quickly detects $p$ to lie in the convex hull.

\item
{\bf Strategy II.}   Given $p,p' \in conv(S)$ and $p$-pivot $v$ in an iteration of the Triangle Algorithm, let $H$ be the orthogonal bisecting hyperplane to the line $pp'$ (see Figure \ref{FigTT2}, consider $p,p',v=v_1$). In the next iteration replace $S$ with $S \cup S'$, where $S'$ is the set of all intersections of the line segments $vv_j$, $v_j \in S$, with the orthogonal hyperplane that bisects $pp'$ (see Figure \ref{FigTT2} $S'=\{v_3', v_4'\}$).

\item
{\bf Strategy III.} Given $p,p',v$, and $H$ as above,  let $\overline H$ be the hyperplane parallel to $H$, passing through $p$.  Let $\overline H_+$ be the halfspace that excludes $p'$.  Let $S'$ be the set of intersections, if any, of the line segments $vv_j$ with $\overline H$ (see Figure \ref{FigTT2},  $\{v_3'', v_4''\}$). Consider the set $conv((S \cap \overline H_+) \cup S')$ (see Figure \ref{FigTT2}, $conv \{ v_1, v_2, v_3'', v_4''\}$, enclosed by the thick lines). Applying the Triangle Algorithm find an approximation, $\overline v$, to the closest point to $p'$ in this set. This is a witness that $p' \not \in conv((S \cap \overline H_+) \cup S')$.
Then we find the closest point to $p$ on the line $p' \overline v$, $\overline p''$. To find $\overline v$ we use the Triangle Algorithm itself. We would expect that $d(p',\overline v) < d(p',v)$ and that the computation of $\overline p''$ can be done efficiently. Instead of using $v$ as pivot at $p'$ to get $p''$, we try to use a better pivot $\overline v$ to get $\overline p''$.

\item
{\bf Strategy IV.}  We label a $v_i \in S$ as cycling pivot if it gets to be selected as a pivot with high frequency.  If in the course of iterations of the Triangle Algorithm we  witness the occurrence of a cycle of pivots with small reduction in the gap, we introduce their corresponding midpoints, or a subset of them, as auxiliary pivots. This will improve their reduction factor in the gap.  Rather than introducing too many auxiliary pivots into $S$,  we can first add the centroid  of the cycling vertices.

\end{itemize}

\begin{figure}[htpb]
	\centering
	\begin{tikzpicture}[scale=0.7]	
      \draw (-4,0) -- (4,0) node[pos=1.0, right] {$\overline H$};
      \draw (-4,-1.5) -- (4,-1.5) node[pos=1.0, right] {$H$};
      \draw (-8,0) -- (-4,0) node[pos=0.55, above] {};
      \draw (-8,-1.5) -- (-4,-1.5) node[pos=0.55, above] {};
      \begin{scope}[red]
      \end{scope}[red]
       \filldraw (-4,0) circle (2pt);
       \filldraw (-4,-3) circle (2pt);
		\draw (-4,0) node[left] {$p$};
\draw (-4,-3) node[left] {$p'$};
\draw (-7,-3.5) node[left] {$v_3$};
\filldraw (-7,-3.5) circle (2pt);
\draw (-7,-3.5) -- (3,1) node[pos=0.5, right] {};
\filldraw (-9,2) circle (2pt);
\draw (3,1) node[right] {$v_1$};
\draw (2,-4) node[right] {$v_4$};
\draw (2,-4) -- (3,1) node[pos=0.5, right] {};
\filldraw (2,-4) circle (2pt);
\draw (-9,2) node[above] {$v_2$};
 \draw [ultra thick] (3,1) -- (-9,2) node[pos=1.0, right] {};
\filldraw (3.,1.) circle (3pt);
\draw (-2.9,1.1) node[right] {$\overline v$};
\filldraw (-3.3,-.15) circle (2pt);
\draw (-3.3,-.15)  node[right] {$\overline p''$};
\draw (-4,-3) -- (-3.,1) node[pos=1.0, right] {};
\draw (-4,0) -- (-3.3,-.15) node[pos=1.0, right] {};
\draw (-4,-3) -- (3.,1) node[pos=1.0, right] {};
\filldraw (-9,2) circle (3pt);
\draw (-4,-3) -- (3,1) node[pos=0.5, right] {};
\draw (-4,0) -- (-2.65,-2.25) node[pos=0.5, right] {};
       \filldraw (-2.65,-2.25) circle (2pt);
		\draw (-4,0) node[left] {$p$};
\draw (-2.65,-2.25) node[right] {$~p''$};
           \draw (-4,0) -- (-4,-3) node[pos=1.0, right] {};
           \draw (-4,-1.5) -- (4,-1.5) node[pos=0.5, right] {};
           \filldraw (-2.55,-1.5) circle (2pt);
           \filldraw (2.5,-1.5) circle (2pt);
           \draw (-2.55,-1.5) node[above] {$v_3'$};
           \draw (2.9,-1.5) node[below] {$v_4'$};
           \filldraw (.82,0) circle (3pt);
           \filldraw (2.8,0) circle (3pt);
           \filldraw (-3,1) circle (2pt);
           \draw (.82,0) node[above] {$v_3''$};
           \draw (3.0,0) node[below] {$v_4''$};
            \draw [ultra thick] (3,1) -- (2.8,0) node[pos=1.0, right] {};
             \draw [ultra thick] (2.8,0) -- (.82,0) node[pos=1.0, right] {};
              \draw [ultra thick] (.82,0) -- (-9,2) node[pos=1.0, right] {};

	\end{tikzpicture}
\begin{center}
\caption{Several examples of auxiliary pivots.} \label{FigTT2}
\end{center}
\end{figure}
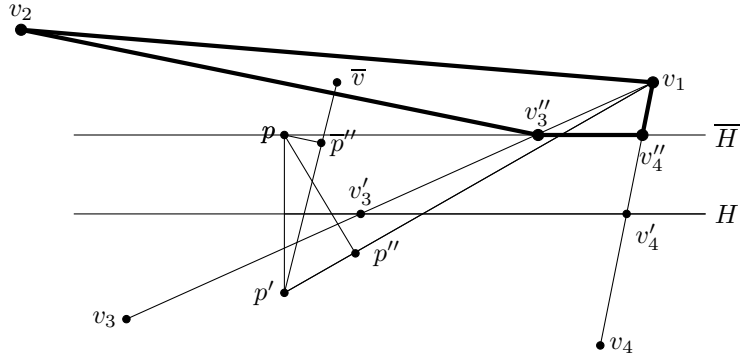

\section{Generalized Triangle Algorithm} \label{sec10}

Let $S$ and $p$ be as before.  In this section we describe generalized Triangle Algorithm where an iterate $p' \in conv(S)$ may be replaced with a convex subset $P'$ of $conv(S)$. First, we prove a theorem that gives a general notion of a pivot, a generalization of Theorem \ref{thm1eq}.

\begin{thm} {\bf (Generalized Distance Duality)} $p \in conv(S)$, if and only if for each closed convex subset $P'$ of $conv(S)$ that does not contain $p$ there exists $v_j \in S$ such that
\begin{equation} \label{genpivot}
d(P',v_j) \geq d(p,v_j).
\end{equation}
\end{thm}

\begin{proof} Suppose $p \in conv(S)$. Since $p \not \in P'$ and $P'$ is closed, there exists $p' \in P'$ such that $p'={\rm argmin}\{d(p,x): x \in P'\}$, and  $p' \not = p$.
Consider the Voronoi cell of $p$ with respect to the two point set $\{p, p'\}$, i.e. $V(p)= \{x \in \mathbb{R} ^m:  \quad d(x,p) < d(x,p')\}$.  We claim that there exists $v_j \in V(p) \cap S$.  If not, $S$ is a subset of  $\overline V(p')= \{x \in \mathbb{R} ^m:  \quad d(x,p) \leq d(x,p')\}$.  By convexity, $\overline V(p')$ contains $conv(S)$.  But since $V(p) \cap \overline V(p')= \emptyset$, this contradicts that $p \in conv(S)$. Now by convexity of $P'$, (\ref{genpivot}) is satisfied.

Conversely, suppose that for any convex subset  $P' \subset conv(S)$ that  does not contain $p$ there exists $v_j \in S$ such that $d(P',v_j) \geq d(p,v_j)$. If $p \not \in conv(S)$, let $p' \in conv(S)$ be the closest point to $p$, $P'=\{p'\}$. Then
$d(P',v_i) < d(p,v_i)$ for all $i$, a contradiction. Thus, $p \in conv(S)$.
\end{proof}

\begin{figure}
\begin{center}
\begin{tikzpicture}
\node [mybox] (box){%
    \begin{minipage}{0.9\textwidth}
{\bf  $\triangle_k$-Algorithm ($S=\{v_1, \dots, v_n\}$, $p$, $\epsilon \in (0,1)$, $k \in \{2, \dots, n\}$)}\
\begin{itemize}

\item
{\bf Step 0.} ({\bf Initialization}) Let $p'=v={\rm argmin}\{ d(p,v_i): v_i \in S\}$. Let
$P'=\{p'\}$.  Set $t=1$.

\item
{\bf Step 1.} If $t=k$, replace $P'$ with $\{p'\}$.  If $d(p,p') < \epsilon d(p,v)$, output $p'$ as $\epsilon$-approximate solution, stop.
It there exists $v_j \in S$ satisfying $d(P', v_j) \geq  d(p, v_j)$, replace $v$ with $v_j$. Otherwise, output $p'$ as a witness, stop.

\item
{\bf Step 2.} Replace $P'$ with $conv(\{P' \cup \{v_j\})$. Replace $t$ with $t+1$.
Compute
\begin{equation}
p'' =\sum_{i=1}^n \alpha'_i v_i = {\rm argmin}\{d(p,x): x \in P'\}.
\end{equation}
Replace $p'$ with $p''$,  $\alpha_i$ with $\alpha'_i$, for all $i=1,\dots,n$. Go to Step 1.
\end{itemize}
    \end{minipage}};
\end{tikzpicture}
 \caption{$\triangle_k$-Algorithm} \label{FigTPPP}
\end{center}
\end{figure}

The theorem asserts that the notion of pivot extends to more general convex subsets of $conv(S)$, suggesting a generalized Triangle Algorithm.  In Figure \ref{FigTPPP} we describe $\triangle_k$-Algorithm.  When $k=2$ the algorithm coincides with the Triangle Algorithm. In each iteration of the algorithm there is a polytope $P'$ that is a convex hull of $t \leq k$ points in $conv(S)$, not containing $p$, and  $p'= {\rm argmin}\{d(p,x): x \in P'\}$ is known. Initially,
$P'=\{p'\}$, a single point.  If $d(p,p')$ is sufficiently small, it stops. Otherwise, it computes a generalized pivot, $v_j \in S$ such that
$d(P',v_j) \geq d(p,v_j)$. To compute such a $v_j$ it suffices to pick a point in
$V(p) \cap S$, where $V(p)$ is the Voronoi region of $p$ with respect to $p'$. It then replaces $P'$ with $conv(P' \cup \{v_j\})$ and updates  $p'$ with $p''= {\rm argmin}\{d(p,x): x \in P'\}$, and repeats the process until $t=k$. It then replaces $P'$ with $\{p'\}$ and repeats the cycle. Thus for $k=2$ the computation of  $p''$ is as in the Triangle Algorithm itself, finding the nearest point to $p$ over a line segment. For $k=3$ the computation of  $p''$ in one cycle amounts to finding the nearest point to $p$, first  over a line segment, then over a triangle. For $k \geq 4$ the minimum distance to $p$ is computed, first  over a line segment, then over a triangle, then over a quadrangle, and so on.

In practice we wish to keep $k$ reasonably small so as to keep the computation of the nearest points sufficiently simple.  As an example Figure \ref{FigTP3} shows one iteration with $k=3$.  In this example, starting with $p'_1$, two iterations of the Triangle Algorithm result in $p'_2$, while a single iteration of $\triangle_3$ results in $p''_2$ which is closer to $p$.

\begin{thm}  Suppose $p \in conv(S)$. For each $ 2 \leq k \leq n-1$, in each iteration $\triangle_{k+1}$-Algorithm produces an approximation to $p$ at least as close as $\triangle_{k}$-Algorithm.  In particular,   $\triangle_{k}$-Algorithm is at least as strong as Triangle Algorithm.  $\Box$
\end{thm}

\begin{remark} The computation of the closest point to $P'$ when it is
the convex hull of $t+1$ vectors in $\mathbb{R}^m$ is a convex programming over a simplex in dimension $t$. However, using the constraint $\alpha_1+ \cdots + \alpha_{t+1}=1$ we can reduce the optimization over a convex quadratic function in $t$ variables over the constraint set
$\{\alpha:  \alpha_1+ \cdots + \alpha_{t} \leq 1, \alpha_i \geq 0, i=1, \dots, t\}$.  In particular, for $t=2$ it is the minimization of a convex quadratic function over
the interval $[0,1]$.  For $t=3$,  the feasible region is the triangle $\{(\alpha_1, \alpha_2):  \alpha_1+ \alpha_2 \leq 1, \alpha_1 \geq 0, \alpha_2 \geq 2\}$.  To minimize a convex quadratic function over a triangle we can first minimize it over the three sides, selecting the best value; then compare this value to  the unconstrained minimum, assuming that it lies feasible to the triangle, and select the best. The complexity progressively gets harder as $t$ increases.  When $k >2 $,  an alternative to computing $p''$ exactly is to find a $p$-witness in $P'$. Such witness gives an approximation of $d(p,P')$ to within a factor of two.
\end{remark}

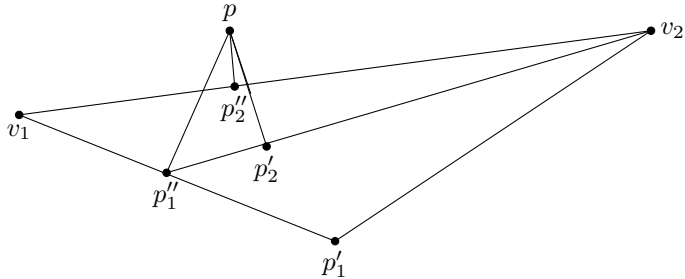
\begin{figure}[htpb]
	\centering
	\begin{tikzpicture}[scale=1.4]	
      \begin{scope}[red]
      \end{scope}[red]
       \filldraw (4,0) circle (1pt);
         \filldraw (0,0) circle (1pt);
       \filldraw (1,-2) circle (1pt);
		\draw (1,-2) node[below] {$p_1'$};
\draw (0,0) node[above] {$p$};
\draw (-2,-.8) -- (1,-2);
\draw (4,0) -- (1,-2);
\draw (0,0) -- (-.6,-1.35);
\draw (-.6,-1.35) node[below] {$p_1''$};
\draw (-2,-.8) node[below] {$v_1$};
\draw (.35,-1.1) node[below] {$p_2'$};
\draw (4,0) -- (-.6,-1.35);
\draw (0,0) -- (.2,-.6);
 \filldraw (.047,-.53) circle (1pt);
 \draw (.047,-.53) node[below] {$p''_2$};;
\draw (0,0) -- (.047,-.53);
\draw (0,0) -- (.35,-1.1);
\draw (-2,-.8) -- (4,0);
\filldraw (-.6,-1.35) circle (1pt);
\filldraw (.35,-1.1) circle (1pt);
\filldraw (-2.,-.8) circle (1pt);
\draw (4,0) node[right] {$v_2$};
	\end{tikzpicture}
\begin{center}
\caption{While two iterations of the Triangle Algorithm would give $p'_2$, $\Delta_2$-Algorithm would result in $p_2''$.} \label{FigTP3}
\end{center}
\end{figure}

For the example when $p$ is at the center of one of the edges of a square (see Figure \ref{FigTP}) the reader can verify that $\triangle_2$-Algorithm solves the problem in one iteration.


\section{Solving LP Feasibility Via The Triangle Algorithm} \label{sec11}

Consider the LP feasibility problem, testing if $\Omega$ is nonempty, where
\begin{equation} \label{omega}
\Omega=\{x \in \mathbb{R} ^n:  \quad Ax=b, \quad x \geq 0 \},
\end{equation}
with $A=[a_1, a_2, \dots, a_n]$, $b, a_i \in \mathbb{R} ^m$. It is well known that through linear programming duality, the general LP problem  can be reduced to a single LP  feasibility problem.  In this section we show how the Triangle Algorithm for problem (P) can be modified to either prove that $\Omega$ is empty, or to compute an approximate feasible point when the set of recession directions of $\Omega$ is empty, where
\begin{equation}
Res(\Omega)=\{d: \quad Ad=0, \quad d \geq 0, \quad d \not =0 \}.
\end{equation}

\begin{definition}  \label{def3} Given $\epsilon$, we  shall say $x_0 \in \mathbb{R} ^n$ is a $\epsilon$-approximate solution (or feasible point) of $\Omega$ if the point $Ax_0$ lies in ${\rm Cone} (\{a_1, \dots, a_n\})$, i.e.
\begin{equation}
Ax_0 \in \{Ax :  \quad x \geq 0 \} ,\end{equation}
satisfying
\begin{equation}
d(Ax_0, b) < \epsilon R',
\end{equation}
where
\begin{equation} \label{rprime3}
R'= \max \{\Vert a_1 \Vert, \dots, \Vert a_n \Vert, \Vert b \Vert\}.
\end{equation}
\end{definition}

The following is easy to show.

\begin{prop} Suppose $0 \not \in conv (\{a_1, \dots, a_n\})= \{Ax:  \quad \sum_{i=1}^nx_i=1, x \geq 0 \}$.  Then  $\Omega \not = \emptyset$ if and only if  $0 \in conv (\{a_1, \dots, a_n, -b \})$. \qed
\end{prop}

\begin{remark}
It is easy to see that $Res(\Omega) = \emptyset$ if and only if $0 \not \in conv (\{a_1, \dots, a_n\})$. In particular, if $Res(\Omega) = \emptyset$, $\Omega$ is a bounded set, possibly empty.  Thus, in this case LP feasibility reduces to a single convex hull decision problem.
\end{remark}

In particular, the above implies we can offer a straightforward variation of the Triangle Algorithm to solve the LP feasibility problem:

\begin{center}
\begin{tikzpicture}
\node [mybox] (box){%
    \begin{minipage}{0.9\textwidth}
{\bf  LP-Feasibility Triangle Algorithm ($A=[a_1, \dots, a_n]$, $b$, $\epsilon \in (0,1))$}\

\begin{itemize}
\item
{\bf Step 0.} ({\bf Initialization})  Let $S=\{a_1, \dots, a_n, a_{n+1}=-b\})$, $p=0$. Let $p'=v={\rm argmin}\{ \Vert a_i \Vert: a_i \in S\}$.

\item
{\bf Step 1.} Given $p'=\sum_{i=1}^{n+1} \alpha_ia_i \in conv(S)$,  if
$\Vert p' \Vert /\alpha_{n+1} < \epsilon \Vert v \Vert$, stop. Otherwise, if there exists a $0$-pivot, replace $v$ with $a_j$. If no $0$-pivot exists, then output $p'$ as a  $0$-witness, stop.

\item
{\bf Step 2.}  Compute the new iterate $p''$ as the nearest point to $p$ on the line segment $p'a_j$.  Replace $p'$ with this point and go to Step 1.
\end{itemize}

    \end{minipage}};
\end{tikzpicture}
\end{center}

The following theorem establishes the needed accuracy to which an approximate solution in $conv (\{a_1, \dots, a_n, -b \})$ should be computed.

\begin{thm} \label{thmc} {\bf (Sensitivity Theorem)} Suppose  $0 \not \in conv (\{a_1, \dots, a_n\})$. Let
\begin{equation} \label{deltaz}
\Delta_0 = \min \bigg \{\Vert p \Vert :  \quad p \in conv (\{a_1, \dots, a_n\}) \bigg \}=\min \bigg \{\Vert Ax \Vert:  \quad \sum_{i=1}^nx_i=1, x \geq 0 \bigg \},
\end{equation}
Let $\Delta_0'$ be any number such that $0 < \Delta_0' \leq \Delta_0$.
Let $b_0= \Vert b \Vert$ and $R'$ as in (\ref{rprime3}).  Suppose $\epsilon>0$ satisfies
\begin{equation}  \label{rprimebd}
\epsilon  \leq \frac{\Delta_0'}{2R'}.
\end{equation}
Suppose we have computed the following approximation to $0$:
\begin{equation} \label{eqthmc1}
p'= \sum_{i=1}^n \alpha_i a_i - \alpha_{n+1} b  \in conv \bigg (\{a_1, \dots, a_n, -b\} \bigg )
\end{equation}
satisfying
\begin{equation} \label{eqthmc2}
\Vert p' \Vert < \epsilon R'
\end{equation}
Let
\begin{equation} \label{eqthmc3}
x_0= \bigg (\frac{\alpha_1}{\alpha_{n+1}}, \dots, \frac{\alpha_n}{\alpha_{n+1}} \bigg )^T.
\end{equation}
Then, $x_0 \geq 0$, and if
\begin{equation} \label{eqthmc3A}
\epsilon'= 2\bigg (1+ \frac{b_0}{\Delta'_0} \bigg ) \epsilon,
\end{equation}
we have
\begin{equation} \label{eqthmc4}
d(Ax_0, b) < \epsilon' R',
\end{equation}
i.e. $x_0$ is an $\epsilon'$-approximate feasible point of $\Omega$.
\end{thm}

\begin{proof} Letting $q'={p'}/{\alpha_{n+1}}$,  from (\ref{eqthmc1}) and (\ref{eqthmc3}) we have
\begin{equation}
q'= \frac{p'}{\alpha_{n+1}} = Ax_0 - b.
\end{equation}
From (\ref{eqthmc2}) we have,
\begin{equation} \label{eqthmc5}
\Vert q' \Vert= d(Ax_0, b) < \frac{\epsilon R'}{\alpha_{n+1}}.
\end{equation}
We wish to compute a lower bound on $\alpha_{n+1}$. Let
\begin{equation}
q= \frac{\alpha_1}{1-\alpha_{n+1}} a_1+ \cdots + \frac{\alpha_n}{1-\alpha_{n+1}} a_n.
\end{equation}
Note that $q \in conv( \{a_1, \dots, a_n\})$.  Then, by definition of $\Delta_0$ we have,
\begin{equation} \label{deltabd}
\Vert q \Vert \geq \Delta_0 \geq \Delta_0'.
\end{equation}
Let
\begin{equation} \label{essec5a}
q''=\frac{p'}{1-\alpha_{n+1}}=q- \frac{\alpha_{n+1}}{1-\alpha_{n+1}}b.
\end{equation}
From (\ref{eqthmc2}) we also have
\begin{equation} \label{essec5b}
\Vert q'' \Vert <  \frac{\epsilon R'}{1-\alpha_{n+1}}.
\end{equation}
 Applying the triangle inequality, $\Vert u \Vert - \Vert v \Vert \leq d(u,v)$, to (\ref{essec5a}) and then using the bound in (\ref{essec5b}) we have
\begin{equation} \label{essec5c}
\Vert q \Vert - \frac{\alpha_{n+1}}{1-\alpha_{n+1}} \Vert b \Vert \leq  \Vert q'' \Vert  < \frac{\epsilon R'}{1-\alpha_{n+1}}.
\end{equation}
From (\ref{essec5c}) and  (\ref{deltabd}), and that $\Vert b \Vert=b_0$ we get
\begin{equation} \label{essec5d}
\Delta_0' < \frac{\alpha_{n+1}}{1-\alpha_{n+1}}b_0 + \frac{\epsilon R'}{1-\alpha_{n+1}}.
\end{equation}
Equivalently,
\begin{equation} \label{essec5e}
\Delta_0' - \epsilon R' < \alpha_{n+1}( \Delta_0' +b_0).
\end{equation}
From  (\ref{essec5e}) and  the assumption in (\ref{rprimebd}) we get
\begin{equation} \label{essec5f}
\alpha_{n+1} > \frac{\Delta_0'- \epsilon R'}{\Delta_0' + b_0} \geq \frac{\Delta_0'}{2(\Delta_0' + b_0)}.
\end{equation}
Substituting the lower bound in (\ref{essec5f}) for $\alpha_{n+1}$ into (\ref{eqthmc5}) implies the claimed error bound in (\ref{eqthmc4}).
\end{proof}

\begin{thm} Suppose $\Omega$ is nonempty and  $Res(\Omega)$ is empty. Given $\epsilon_0 \in (0,1)$, in order to compute $x_0 \geq 0$ such that $d(Ax_0,b) < \epsilon_0 R'$ it suffices to compute a point $p' \in conv(\{a_1, \dots, a_n, -b\}$ so that
\begin{equation} \label{ebound}
\Vert p' \Vert < \epsilon R', \quad  \epsilon \leq \frac{\Delta_0'}{2}\min \bigg \{ \frac{1}{R'}, \quad \frac{\epsilon_0}{(\Delta_0' + b_0)} \bigg \},\end{equation}
where $\Delta_0'$ is any number satisfying $0 < \Delta_0' \leq \Delta_0$. The number of arithmetic operations of the Triangle Algorithm is
\begin{equation}
O\bigg (mn \min \bigg \{\frac{{R'}^2}{{\epsilon_0^2 \Delta'_0}^2}, \frac{1}{c'}\ln \frac{R'}{\epsilon_0 \Delta'_0} \bigg \} \bigg), \quad c'=c \bigg (0,[A,-b], \frac{\epsilon_0 \Delta'_0}{4R'} \bigg ) \geq  \frac{{\epsilon_0^2 \Delta'_0}^2} {{4R'}^2}.
\end{equation}
\end{thm}

\begin{proof} The upper bound on $\epsilon$  in (\ref{ebound}) follows from the sensitivity theorem.  Since $\Delta_0 \leq R'$,  from this upper bound  it suffices to pick $\epsilon  \leq {\Delta'_0}\epsilon_0/{4R'} $. Then the claimed complexity for computing an $\epsilon_0$-approximate solution follows from Theorem \ref{compthm}.
\end{proof}

\begin{remark}
In theory, to estimate the number of needed iterations requires an estimate $\Delta_0'$, a lower bound  to $\Delta_0$. However, in practice given a prescribed accuracy $\epsilon_0$ we merely need to run the Triangle Algorithm until either we have computed an $\epsilon_0$-approximate solution of $\Omega$, or a $p$-witness proving that it is empty. Despite this alternative, there is a way to get an estimate of $\Delta_0$ as described in the next remark.
\end{remark}

\begin{remark} Since $\Delta_0$ is unknown, in practice we can use an estimate.  One possible approach is first to try to test if $0$ lies in $conv(\{a_1, \dots, a_n\})$. Assuming that $\Omega$ has no recession direction, $0$ is not in this convex hull. Thus by applying the Triangle Algorithm we will get a $p$-witness $p'$ such that $d(p', a_i) < \Vert p' \Vert$ for all $i=1, \dots, n$. Such a point $p'$ by Corollary \ref{cornew}  will necessary satisfy:
\begin{equation}
\frac{1}{2} \Vert p' \Vert  \leq \Delta_0 \leq \Vert p' \Vert.
\end{equation}
Thus by solving this auxiliary convex hull decision problem we get a $p$-witness, $p'$, whose norm can be used as $\Delta_0'$.  Next we use $p'$ as the starting iterate as it already lies in the $conv(\{a_1, \dots, a_n, -b\})$.
\end{remark}

Following the above remark we offer a two-phase Triangle Algorithm for solving the feasibility problem in LP with the assumption that $0 \not \in conv (\{a_1, \dots, a_n\})$:

\begin{center}
\begin{tikzpicture}
\node [mybox] (box){%
    \begin{minipage}{0.9\textwidth}
{\bf  Two-Phase LP-Feasibility Triangle Algorithm ($A=[a_1, \dots, a_n]$, $b$, $\epsilon_0 \in (0,1)$)}\

\begin{itemize}
\item
{\bf Phase I.}
\begin{itemize}
\item

{\bf Step 1.} Call {\bf Triangle Algorithm}($S=\{a_1, \dots, a_n\}$, $p=0$, $\epsilon=0.5$).

\item

 {\bf Step 2.} If the output $p'$ is not a witness, replace $\epsilon$ with $\epsilon/2$, go to Step 1.
 \end{itemize}

\item
{\bf Phase II.} Let $\Delta_0'=0.5 \Vert p' \Vert$,  $p'$ the $0$-witness output in Phase I. Set
\begin{equation}
\epsilon_1= \frac{\Delta_0'}{2}\min \bigg \{ \frac{1}{R'}, \quad \frac{\epsilon_0}{(\Delta_0' + b_0)} \bigg \}.\end{equation}
Call {\bf Triangle Algorithm} ($S=\{a_1, \dots, a_n, -b\}$, $p=0$, $\epsilon= \epsilon_1$).
\end{itemize}

    \end{minipage}};
\end{tikzpicture}
\end{center}

From the sensitivity theorem and the complexity theorem, Theorem \ref{compthm}, the complexity of the Two-Phase LP-Feasibility Triangle Algorithm can be stated as

\begin{thm} Suppose $\Omega$ is nonempty and  $Res(\Omega)$ is empty. Given $\epsilon_0 \in (0,1)$, in order to compute an $\epsilon_0$-approximate solution of $\Omega$ (i.e. $x_0 \geq 0$ such that $d(Ax_0,b) < \epsilon_0 R'$), it suffices to set  $\Delta'_0=0.5 \Vert p' \Vert$, where $p'$ is the $p$-witness computed in Phase I.  Then in Phase II it suffices to compute a point $p' \in conv(\{a_1, \dots, a_n, -b\}$ so that
\begin{equation}
\Vert p' \Vert < \epsilon R', \quad  \epsilon \leq \epsilon_1= \frac{\Delta_0'}{2}\min \bigg \{ \frac{1}{R'}, \quad \frac{\epsilon_0}{(\Delta_0' + b_0)} \bigg \}.\end{equation}

The number of arithmetic  operation of Phase I is
\begin{equation}
O\bigg (mn \min \bigg \{ \frac{R'^2}{\Delta_0^2}, \frac{1}{c_0}\ln \frac{\delta_0}{\Delta_0} \bigg \} \bigg ),
\end{equation}
where $c_0=c(0,A,{\Delta_0}/{R'}) \geq {\Delta^2_0}/{R'^2}$ is visibility factor of $0$ with respect to $A$. The number of arithmetic  operation of Phase II is
\begin{equation}
O \bigg (mn \min \bigg \{ \frac{R'^2}{\Delta_0^2\epsilon_0^{2}}, \frac{1}{c_1}\ln \frac{\delta_0}{\Delta_0 \epsilon_0} \bigg \} \bigg ),
\end{equation}
where $c_1=c(0,[A,-b], \frac{\epsilon_0 \Delta_0}{4R'}) \geq (\frac{\epsilon_0 \Delta_0}{4R'})^2$ is visibility factor of $0$ with respect to $[A,-b]$. $\Box$
\end{thm}

\section{Solving General LP Feasibility Via Triangle Algorithm} \label{sec12}
Here again we consider the problem of computing $\epsilon$-approximate solution to
$\Omega=\{x \in \mathbb{R} ^n:  \quad Ax=b, \quad x \geq 0 \}$, is one exists (see (\ref{def3})). Whether or not $0 \in conv (\{a_1, \dots, a_n\})$, it is well known that to test the feasibility of $\Omega$  one can safely add the constraint $\sum_{i=1}^n x_i \leq M$, where $M$ is a large enough constant. For integer inputs,  such $M$ can be computed, dependent on the size of encoding of $A,b$, see e.g. Schrijver \cite{sch86}.  In fact it can be shown that $M$ can be taken to be $O(2^{O(L)})$, where $L$ dependents on $m,n$ and logarithm of the absolute value of the largest entry of $A$ or $b$, see e.g. \cite{kal97}.

Having such a bound $M$, by adding a slack variable, $x_{n+1}$ to this constraint, testing the feasibility  of $\Omega$ is equivalent to testing the feasibility of
 \begin{equation}
 \Omega_M= \bigg \{(x, x_{n+1}) \in \mathbb{R} ^{n+1}:  \quad Ax=b, \quad \sum_{i=1}^{n+1} x_i = M, \quad x \geq 0, \quad x_{n+1} \geq 0  \bigg \}.
 \end{equation}
Dividing the equations in $\Omega_M$ by $M$ and setting $y_i=x_i/M$, $i=1, \dots, n+1$, we conclude that testing the feasibility of $\Omega_M$ is equivalent to testing the feasibility of
\begin{equation}
\overline \Omega_M= \bigg \{ (y, y_{n+1}) \in \mathbb{R} ^{n+1}: \quad  Ay= \frac{b}{M}, \quad \sum_{i=1}^{n+1} y_i = 1, \quad y \geq 0, \quad y_{n+1} \geq 0 \bigg \}.
\end{equation}
We call the {\it augmented} (P) the problem of testing if
$p \in  conv(\overline S)$, where

\begin{equation} \label{augmented}
p=\frac{b}{M}, \quad \overline S= \bigg \{a_1, \dots, a_n, a_{n+1}  \bigg \}, \quad a_{n+1}=0.
\end{equation}

We may solve the augmented (P) in two different ways, directly by solving a single convex hull decision problem, or by solving a sequence of such problems. We will describe the two approaches and analyze their complexities. First, we prove an auxiliary lemma, an intuitively simple geometric result.

\begin{lemma} \label{lem2} Given  $u, w \in \mathbb{R} ^m$,  we have:
\begin{equation}
\sup \bigg \{d \bigg (u,\frac{w}{\mu} \bigg ): \quad \mu \in [1, \infty) \bigg \} = \max \bigg \{d(u, w), \Vert u \Vert \bigg \}.
\end{equation}
\end{lemma}

\begin{proof} Consider the maximum of the function $f(t)= d^2(u, t w)$
over the interval $t \in [0,1]$. Since $f(t)$ is convex, its maximum is attained at an endpoint of the interval. Letting $t = 1/\mu$, the proof is complete.
\end{proof}

\begin{thm} \label{thm5} Suppose that a positive number $M$ satisfies the property that $\Omega$ is feasible if and only if $\Omega_M$ is feasible.
Then, solving the augmented (P) to within accuracy of $\epsilon/ M$  is equivalent to testing the  feasibility of $\Omega$ to within accuracy of $\epsilon$.  More specifically, by applying the Triangle Algorithm to solve the augmented (P) {\rm (see (\ref{augmented}))}, in $O(mn \min \{M^2\epsilon^{-2}, M(c''\epsilon)^{-1}\})$ arithmetic operations, where  $c''=c(bM^{-1},[A, 0], \epsilon) \geq M^2\epsilon^{-2}$,
either we compute a $p$-witness $p' \in conv(\overline S)$ proving that $b/M \not \in \overline S$ (hence $\Omega = \emptyset$), or a point  $p' \in conv(\overline S)$
such that for some $i=1, \dots, n+1$ we have,
\begin{equation} \label{eqLP1}
d(b,Mp') < \epsilon  \max \bigg \{d(b, a_i), \Vert a_i \Vert \bigg \} < 2 R' \epsilon,
\end{equation}
where $R'$ is as in (\ref{rprime3}). Equivalently, $p'=Ay_0$, for some $y_0 \geq 0$, and if $x_0=My_0$, then
\begin{equation} \label{eqLP2}
d(b,Ax_0) < \epsilon  \max \bigg \{d(b, a_i), \Vert a_i \Vert \bigg \} < 2 R' \epsilon.
\end{equation}
\end{thm}

\begin{proof} If solving the augmented (P) to accuracy $\epsilon/M$ leads to a $p$-witness $p' \in  conv(\overline S)$ proving that $b/M \not \in  \overline S$, then $\Omega$ is empty.  Otherwise, by Theorem \ref{thm4} the algorithm leads to a point  $p' \in conv(\overline S)$ such that for some $i=1, \dots, n+1$ we have
\begin{equation}
d \bigg (\frac{b}{M},p' \bigg ) <  \frac{\epsilon}{M}  d \bigg (\frac{b}{M}, a_i \bigg).
\end{equation}
Multiplying the above by $M$, applying Lemma \ref{lem2},  and since $M \times d(b/M,p')=d(b,Mp')$, we get the proof of the first claimed inequalities in
(\ref{eqLP1}) and (\ref{eqLP2}).  The bound $2R' \epsilon$ follows from the triangle inequality, $d(b,a_i) \leq \Vert b \Vert + \Vert a_i \Vert$. The claimed complexity bound follows from Theorem \ref{compthm}.
\end{proof}

Instead of solving the augmented (P) with an a priori estimate for $M$, we may solve it as a sequence of augmented (P)'s with increasing estimates of $M$ that successively doubles in value. To describe this approach,  first consider the following.\\

Given $\mu >0$, let $p_\mu=b/\mu$. Select any $p' \in conv(\overline S)$ as the iterate and let
\begin{equation}
\mu_0= \sup  \bigg \{\mu:  \quad d(p_\mu,a_i) >  d(p',a_i), \quad \forall i=1, \dots, n+1 \bigg \}.
\end{equation}
Clearly, $0 < \mu_0 < \infty$. According to Theorem \ref{thm2},  and the definition of $\mu_0$, for any $\mu \in (0, \mu_0)$,  $p'$ is a $p$-witness proving that $p_\mu \not \in conv(\overline S)$. From Lemma \ref{lem2}, we know $\mu_0 \geq 1$.  Given $\epsilon>0$, Set $\mu=\mu_0$ and consider the following iterative step:\\

\begin{center}
\begin{tikzpicture}
\node [mybox] (box){%
    \begin{minipage}{0.9\textwidth}
{\bf Iterative Step.}  Let $\epsilon_\mu = \epsilon /\mu$.  Use the Triangle Algorithm to solve the augmented (P) with $p_\mu=b/\mu$ to
either compute $p'_\mu \in conv(\overline  S)$ such that
\begin{equation}
d(p_\mu, p'_\mu) < \epsilon_\mu  \max \bigg \{d(p_\mu, a_i), \Vert a_i \Vert \bigg \}, \quad \text{for some}~i,
\end{equation}
or a $p$-witness $p'_\mu \in conv(\overline  S)$ proving that
$p_\mu \not \in S$.  In the latter case replace $\mu$ with $2\mu$ and repeat.

    \end{minipage}};
\end{tikzpicture}
\end{center}

The following theorem analyzes the complexity of this approach. However, we only consider it according to the first complexity bound.  An alternate complexity bound using visibility factor can also be given.

\begin{thm} \label{thm6}
Repeating the iterative step, either we compute an $\epsilon$-approximate feasible point of $\Omega$, or a $p$-witness to the infeasibility of the augmented (P) with $\mu  \geq M$. In the latter case $\Omega$ is empty. More specifically,  if the algorithm requires $r$ iterations of the iterative step,  its arithmetic complexity is
$O(R'^2mn 2^{2r} \epsilon^{-2})$. Furthermore, $r \leq  \lceil \log_2 M \rceil$.
\end{thm}

\begin{proof}  Since $\mu_0 \geq 1$, the number of augmented (P)'s  to be solved is bounded by $t=\lceil \log_2 M \rceil$.  With the initial value of $\mu=\mu_0$ we solve the augmented (P) to either compute an approximate solution  in $\Omega$ to within accuracy of $\epsilon$, or a $p$-witness to the infeasibility of $\Omega_\mu$.  By Theorem \ref{thm4} this takes $O(\epsilon^{-2})$ arithmetic operations.  If $\Omega_\mu$ is infeasible, we double $\mu$ and test the feasibility of new $\Omega_\mu$. Then by (\ref{iter}) in  Lemma \ref{lem1}, the number of iterations needed to halve the current error that is known to exceed $\epsilon$ is bounded by
\begin{equation}
(32 \ln 2) \frac{\overline R^2}{\epsilon^2}, \quad \overline R = \max \bigg \{ \max \{d(b, a_i), \Vert a_i \Vert \}:  \quad i=1, \dots, n+1 \bigg \} \leq 2 R',
\end{equation}
where the bound $2R'$ is from the fact that $d(b,a_i) \leq \Vert b \Vert + \Vert a_i \Vert$,
Repeating this $r$ times we get the total complexity bounded by
\begin{equation}
\overline K_\epsilon \leq    (32 \ln 2) \frac{\overline R^2}{\epsilon^2} (1+2^2+2^4 + \dots + 2^{2r}) \leq   (32 \ln 2) \frac{\overline R^2}{\epsilon^2} 2^{2r+1} =
(64 \ln 2) \frac{\overline R'^2}{\epsilon^2} 2^{2r+1}
= O \bigg (R'^2\frac{2^{2r}}{\epsilon^2} \bigg).\end{equation}
\end{proof}

{\bf Concluding  Remarks and Future Work.}   In this article we have described several novel characterization theorems together with a very simple algorithm for the convex hull decision problem (problem (P)), the {\it Triangle Algorithm}.  The Triangle Algorithm can also be considered as an algorithm that tests if a given set of open balls in $\mathbb{R}^m$ that are known to have a common point on their boundary, have a nonempty intersection. The Triangle Algorithm is straightforward to implement and it is quite flexible in the sense that it allows variations to improve practical performance. Its main step is the comparison of distances in identifying a good $p$-pivot or a $p$-witness.  While the Triangle Algorithm works with distances, it requires only the four elementary operations and comparisons, no square-root operation is required.  In the worst-case, each iteration requires $O(mn)$ arithmetic operations. However, when $p \in conv(S)$,  the algorithm may be able to find a good pivot by searching a constant number of $v_i$'s. In this case the complexity of each iteration is $O(m +n)$  operations ($O(m)$ to find a pivot and $O(n)$ to update the coefficients in the representation of the pivot as a convex combination of $v_i$'s).

In this article we have also made some theoretical comparisons with such algorithms as the simplex method, Frank-Wolfe method, and first-order gradient algorithms.  In a forthcoming article we will make actual computation comparisons  with some of these algorithms, indicating that the Triangle Algorithm is quite competitive with these algorithms when solving the convex hull decision problem.  We emphasize that while the Triangle Algorithm computes an iterate so that the relative error is within prescribed tolerance, it can produce an approximation with arbitrarily small absolute error, possibly at the cost of a theoretical complexity with a larger constant. Such is the case with all approximation schemes, as well as polynomial-time algorithms.  Analogous to  polynomial-time algorithms for linear programming, when the input data is integer, a sufficient approximation can be rounded into an exact representation of $p$ as a convex combination of $v_i$'s. When the rank of the matrix of points in $S$ is $m$, it is possible to represent each iterate as a convex combination of at most $m+1$ of the $v_i$'s. It maybe convenient to maintain such representation within each iteration.
While the Triangle Algorithm is designed to solve problem (P), we have analyzed its theoretical applicability in solving LP feasibility, as well as LP optimization.

We anticipate that the simplicity of the algorithm, its theoretical properties, and the distance dualities will lead to new analysis and applications.  For instance, computational testing, average case analysis, amortized complexity, analysis of visibility constant for
special problems, generalization of the characterization theorems and the Triangle Algorithm, as well as its specializations.  In addition to the convex hull decision problem itself, some applications are straightforward. For instance, the Triangle  Algorithm when applied to each $p=v_i$ with $S_i=S- \{v_i\}$ gives an approximation algorithm for solving the irredundancy problem.  Some other applications are not so straightforward.  In fact since the release of the first version of this article, \cite{kal12},  we have considered several novel applications or extensions of the Triangle Algorithm.

(i) In \cite{kal12a} we describe the application of the  Triangle Algorithm in order to solve a linear system.  As such it offers alternatives to iterative methods for solving a linear system, such as Jacobi, Gauss-Sidel, SOR, MOR, and others.  We will also report on the computational performance of the Triangle Algorithm with such algorithms.

(ii) In \cite{Kalan12} we give a version of the Triangle Algorithm for the case of problem (P) where we wish to locate the coordinates of a point $p$ having prescribed distances to the sites $S=\{v_1, \dots, v_n\}$. We call this the {\it ambiguous convex hull problem} since not only the coordinates of the point are unknown, so is the existence of such point.  Despite this ambiguity, a variation of the Triangle Algorithm can solve the problem in $O(mn\epsilon^{-2} \ln \epsilon^{-1})$ arithmetic operations.

(iii): In \cite{{kal13a}} we will give a generalization of the Triangle Algorithm that either separate two compact convex subsets, or computes an approximate point in their intersection. It can also approximate the distance between them when they are distinct, or compute supporting hyperplanes with the largest  margin.

In addition to what is stated above, other applications and generalizations of the characterization theorems, as well as the Triangle Algorithm itself are possible and will be considered in our future work.  These include, specialization of the Triangle Algorithm to LP  with special structures, combinatorial and graph optimization problems.
For instance, in another forthcoming article we will analyze the Triangle Algorithm for the cardinality matching in a bipartite graph, showing that it exhibits polynomial-time complexity. As such, the algorithm is very different than all the existing polynomial-time algorithms for this classic problem. As generalizations, it is possible to state problem (P) and the corresponding characterizations when the set $S$ consists of one or a finite number of compact subsets of $\mathbb{R} ^m$, e.g.   polytopes, balls, or more general convex bodies.  Moreover, problem (P) can also be defined over other cones.  For instance,  a canonical problem in semidefinite programming described in \cite{kal2003} is an analogue of problem (P) over the cone of positive semidefinite matrices. Among other interesting research topics, the computation of visibility constant such as $\nu^*$  (see (\ref{nustar})) and $\nu_*$ (see (\ref{nustar1})), and $\lambda_*$ (see (\ref{lambda})) for some special cases of the convex hull problem could result in proving very efficient or polynomial-time complexity. In future work we hope to investigate this problem as well.\\

{\bf Acknowledgements.} I thank Iraj Kalantari for a discussion on the use of Voronoi diagram argument to prove Theorem \ref{thm1}, resulting in a simpler geometric proof than the one given in an earlier version of this article. I like to thank Tong Zhang for a discussion regarding the Greedy Algorithm for general convex minimization over a simplex. I also thank G\"unter Rote for bringing to my attention the work of Kuhn \cite{kuhn} and Durier and Michelot \cite{DM86} (Remark \ref{rm0}).


\bigskip



\begin{thebibliography}{9}
\bibitem{kal13a} B. Kalantari, An algorithmic separating hyperplane theorem and its applications, forthcoming. \filbreak
\end{thebibliography}
\end{document}